\documentclass[a4paper,UKenglish,cleveref, autoref, thm-restate, pdfa]{lipics-v2021}

\pdfoutput=1 
\hideLIPIcs  

\usepackage{graphicx}
\usepackage{amsmath}
\usepackage{algorithm}
\usepackage{algpseudocode}
\usepackage{url}
\usepackage{booktabs}
\usepackage{array}
\usepackage{multirow}
\usepackage{tikz}

\title{Faster Linear-Space Data Structures for Path Frequency Queries} 

\author{Ovidiu Rața}{University of Copenhagen, Copenhagen, Denmark}{ovidiurata03@gmail.com}{https://orcid.org/0009-0006-8816-2478}{}

\authorrunning{O. Rața} 

\Copyright{Ovidiu Rața} 



\category{} 

\relatedversion{} 

\funding{Supported by VILLUM Foundation grant 54451, Basic Algorithms Research Copenhagen (BARC).}

\acknowledgements{I wish to thank Professor Mikkel Thorup for his valuable support as my Master's Thesis supervisor. I also wish to thank the anonymous reviewers of SWAT 2026 for their meaningful feedback, which helped improve the paper.}

\nolinenumbers 

\EventEditors{Pierre Fraigniaud}
\EventNoEds{1}
\EventLongTitle{20th Scandinavian Symposium on Algorithm Theory (SWAT 2026)}
\EventShortTitle{SWAT 2026}
\EventAcronym{SWAT}
\EventYear{2026}
\EventDate{June 17--19, 2026}
\EventLocation{Copenhagen, Denmark}
\EventLogo{}
\SeriesVolume{370}
\ArticleNo{15}

\begin{document}

\maketitle

\begin{abstract}
 We present linear-space data structures for several frequency queries on trees, namely: path mode, path least frequent element, and path $\alpha$-minority queries. We present the first linear-space data structures, requiring $O(n \sqrt{nw})$ preprocessing time, that can answer path mode and path least frequent element queries in $O(\sqrt{n/w})$ time. This improves upon the best previously known bound of $O(\log\log n \sqrt{n/w})$ achieved by Durocher et al.~\cite{durocher2016linear} in 2016. 
 
 For the path $\alpha$-minority problem, where $\alpha$ is specified at query time, we reduce the query time of the linear-space data structure of Durocher et al.~\cite{durocher2016linear} from $O(\alpha^{-1}\log\log n)$ down to $O(\alpha^{-1})$ by employing a simple randomized algorithm with a success probability $\geq 1/2$. 
 
 We also present the first linear-space data structure supporting ``Path Maximum $g$-value Color'' queries in $O(\sqrt{n/w})$ time, requiring $O(n \sqrt{nw})$ preprocessing time. This general framework encapsulates both path mode and path least frequent element queries. For our data structures, we consider the word-RAM model with $w\in \Omega(\log n)$, where $w$ is the word size in bits.

\end{abstract}

\section{Introduction}
\label{section:introduction}
According to Chan et al.~\cite{chan2014linear}, given a multiset $S$, the frequency of an element $x$ in $S$, denoted by $\text{freq}_S(x)$, is the number of occurrences of the element $x$ in $S$. A mode of $S$ is an element $a\in S$ such that for all $x\in S$ , $\text{freq}_S(x)\leq\text{freq}_S(a)$. Given an array $A$ of $n$ integers, a range mode of $A[i:j]$ is a mode of the multiset determined by the range $A[i:j]$. Alternatively, one least frequent element in a multiset $A$ is an element of minimum multiplicity in $A$.
    
    Generalizing the range mode query problem to trees, we study the problem of building a data structure that, given two nodes $(i,j)$, can efficiently find the mode of the multiset of nodes of the path $P(i,j)$. An alternative to the path mode query is the path least frequent element query, which asks to report \textbf{any single} element on the path $P(i,j)$ of minimal multiplicity. 
    
    Historically, there has always been a gap between the query times for the array and tree versions of the mode and least frequent element problems. Chan et al.~\cite{chan2014linear} developed the fastest known linear-space data structure with $O(\sqrt{n/w})$ query time for the array mode query problem, while the time for tree path queries remained $O(\log n\sqrt{n})$. Later, in 2016, Durocher et al.~\cite{durocher2016linear} reduced the time on trees down to $O(\log\log n \sqrt{n/w})$, for both mode and least frequent element queries. In this paper, we bridge this gap by reducing the time for both path mode queries and least frequent element queries on trees down to $O(\sqrt{n/w})$, in the context of linear-space data structures. 
    
    We also consider a new, more general problem, which we call the \textbf{Range Maximum $g$-value Color} query. By abstracting specific frequency conditions into a generic function $g$, we provide a unified algorithmic framework that simultaneously solves both the mode and least frequent element problems, while also capturing more complex frequency-based metrics. An array $A$ of $n$ positive integers (the ``array of colors'') and a function $g$ are given. We require $g$ to satisfy several properties, which are also satisfied by the range mode function. A detailed formal definition of these properties is provided in Section~\ref{section:problem_presentation}.
    
    The \textbf{Range Maximum $g$-value Color} problem is to preprocess $A$ and the function $g$ to efficiently answer queries of the following form:
    \begin{quote}
        Given two indices $i, j$ of $A$, determine: $\arg\max_{c \in A[i \dots j]} \{ g(i, j, c) \}$
    \end{quote}
    In other words, the query asks to report any single color $c$ corresponding to the maximum value of $g(i,j, c)$.
    
    We study the generalization of this problem to trees, which we call \textbf{Path Maximum $g$-value Color}, and prove that there exists a linear-space data structure supporting these queries in $O(\sqrt{n/w})$ time, requiring $O(n \sqrt{nw})$ preprocessing time. Because this problem acts as a generalization, it directly reduces the $O(\log\log n \sqrt{n/w})$ worst-case query time required by the previous linear-space data structures for path mode and path least frequent element queries~\cite{durocher2016linear}, down to $O(\sqrt{n/w})$.

    We conclude our work by presenting an improvement over the path $\alpha$-minority problem, studied by Durocher et al.~\cite{durocher2016linear}. The path $\alpha$-minority query for a tree path $(u,v)$ asks to report \textbf{any single} color that occurs at most an $\alpha$ fraction of the time on the path between nodes $u$ and $v$. We prove that by applying a simple randomized algorithm with success probability $\geq 1/2$, their linear-space data structure can support $\alpha$-minority queries in $O(\alpha^{-1})$ time (where $\alpha$ is specified at query time), improving on the previous $O(\alpha^{-1}\log\log n)$ query time.
    
    We assume the Word RAM model of computation using words of $w =\Omega(\log n)$ bits, where $n=|A|$.
    
    \subsection{Related Work and Contribution}

    Krizanc et al.~\cite{krizanc2005range} presented $O(n)$-space data structures that support range mode queries in $O(\sqrt{n}\log \log n)$ time on arrays and $O( \sqrt{n} \log n)$ time on trees. Chan et al.~\cite{chan2014linear} achieved $o(\sqrt{n})$ query time with an $O(n)$-space data structure that supports queries in $O(\sqrt{n/w} ) \subseteq O(\sqrt{n/ \log n})$ time on arrays. Durocher et al.~\cite{durocher2016linear} present an $O(n)$ space data structure supporting path mode queries on trees in $O(\log\log n \sqrt{ n/w })$ time. Durocher et al.~\cite{durocher2016linear} ask whether the $\log\log n$ gap between the path mode query time on trees and the range mode query time on arrays can be closed. In this paper, we present the first $O(n)$ space data structure that supports path mode queries in $O(\sqrt{ n/w })\subset o(\log\log n \sqrt{ n/w })$ query time, thus eliminating the gap between path mode and range mode query times. 
    
    In the context of the least frequent element problem, Chan et al.~\cite{chan2012linear} presented an $O(n)$-space data structure that supports range least frequent element queries on arrays in $O(\sqrt{n})$ time. Durocher et al.~\cite{durocher2016linear} improve the range least frequent element query time down to $O(\sqrt{n/w})$, and designed the first path least frequent element query data structure for trees, with an $O(\log\log n\sqrt{n/w})$ query time. In our paper, we present the first linear-space data structure that supports path least frequent element queries on trees in $O(\sqrt{n/w})$ time, thus also closing the $\log \log n$ gap between the array and tree query times for this problem. A history of previous best query times, achieved by linear-space data structures for mode queries and least frequent element queries, is presented in Tables~\ref{tab:mode_history},~\ref{tab:lfe_history}.  

    \begin{table}[H]
    \centering
    \caption{History of Mode Queries (linear-space)}
    \begin{tabular}{lccc}
        \toprule
        \textbf{Authors} & \textbf{Array Query Time} & \textbf{Tree Query Time} & \textbf{Source} \\
        \midrule
        Krizanc et al. & $O(\sqrt{n} \log \log n)$ & $O(\sqrt{n} \log n)$ &~\cite{krizanc2005range} \\
        Chan et al. & $O(\sqrt{n/w})$ & -- &~\cite{chan2014linear} \\
        Durocher et al. & -- & $O(\log \log n \sqrt{n/w})$ &~\cite{durocher2016linear} \\
        \textbf{This Paper} & -- & $\mathbf{O(\sqrt{n/w})}$ & -- \\
        \bottomrule
    \end{tabular}
    \label{tab:mode_history}
\end{table}

\begin{table}[H]
    \centering
    \caption{History of Least Frequent Element Queries (linear-space)}
    \begin{tabular}{lccc}
        \toprule
        \textbf{Authors} & \textbf{Array Query Time} & \textbf{Tree Query Time} & \textbf{Source} \\
        \midrule
        Chan et al. & $O(\sqrt{n})$ & -- &~\cite{chan2012linear} \\
        Durocher et al. & $O(\sqrt{n/w})$ & $O(\log \log n \sqrt{n/w})$ &~\cite{durocher2016linear} \\
        \textbf{This Paper} & -- & $\mathbf{O(\sqrt{n/w})}$ & -- \\
        \bottomrule
    \end{tabular}
    \label{tab:lfe_history}
\end{table}

    In their work, Durocher et al.~\cite{durocher2016linear} present a linear-space data structure that can answer path $\alpha$-minority queries on trees in $O(\alpha^{-1}\log \log n)$ time, with $\alpha$ specified at query time. The path $\alpha$-minority query for a tree path $(u,v)$ asks to report any color that occurs at most an $\alpha$ fraction of the tree path between nodes $u$ and $v$.  We prove that by applying a simple randomized algorithm with success probability $\geq 1/2$, their linear-space data structure can support $\alpha$-minority queries in $O(\alpha^{-1})$ time, improving on the previous $O(\alpha^{-1}\log\log n)$ query time. A history of query times achieved by linear-space data structures for $\alpha$-minority is given in Table~\ref{tab:minority_history}.

\begin{table}[H]
    \centering
    \caption{History of $\alpha$-Minority Queries (linear-space)}
    \begin{tabular}{lccc}
        \toprule
        \textbf{Authors} & \textbf{Array Query Time} & \textbf{Tree Query Time} & \textbf{Source} \\
        \midrule
        Chan et al. & $O(1/\alpha)$ & -- &~\cite{chan2012linear} \\
        Durocher et al. & -- & $O(\alpha^{-1} \log \log n)$ &~\cite{durocher2016linear} \\
        \textbf{This Paper} & -- & $\mathbf{O(1/\alpha)}$ & (Randomized) \\
        \bottomrule
    \end{tabular}
    \label{tab:minority_history}
\end{table}

\section{Problem Presentation}
\label{section:problem_presentation}
We begin with a detailed formalization of the main problem studied in this paper, namely, the \textbf{Path Maximum $g$-value Color} problem. We formulated this problem to generalize the path mode query problem as broadly as possible. We also note that we present the first efficient linear-space data structure for the range version of the maximum $g$-value color problem, as this general problem has not been studied previously.

Our study of this problem and the methods used to solve it were inspired by the domain of retrieval/counting problems and array range query problems. A description of retrieval/counting problems is given by Chazelle~\cite{chazelle1986filtering}, and some interesting instances of array range query problems can be found in Skala et al.~\cite{Skala2013}. For clarity, we first present the problem in the context of arrays, and subsequently generalize it to trees.

\subsection{The Range Maximum \texorpdfstring{$g$}{g}-value Color Problem}

An instance of our problem consists of an array of integers $A$ (referred to as the ``array of colors'') and an integer function $g$. Let $n = |A|$. We assume that $1\leq A[i]\leq n$. If $A$ does not fulfill this requirement, $A$ can be transformed in $O(n\log n)$ time to an array $A'$ with values compressed to the interval $\{1,\dots,n\}$. 

The starting point of our study was the range mode function. We developed our data structure using only a subset of the properties of the range and path mode functions. We require $g$ to satisfy the following constraints, all of which are shared by the range mode function:

\begin{enumerate}
    \item \textbf{Domain:} The function $g(i, j, c)$ takes three integer arguments, where $1 \leq i \leq j \leq n$ are indices of $A$, and $c$ is a color present in the subarray $A[i \dots j]$.
    \item \textbf{Output:} $g$ produces an integer output that fits into a single machine word.
    \item \textbf{Range Contraction:} For any valid input $(i, j, c)$,
    \[
        g(i, j, c) = g(i', j', c)
    \]
    where $i'$ and $j'$ are the minimum and maximum indices, respectively, such that $i \leq i' \leq j' \leq j$ and $A[i'] = A[j'] = c$. In other words, the value of $g$ depends only on the range spanned by the occurrences of $c$ within the query interval.
    \item \textbf{Oracle Complexity:} We assume $g$ can be computed efficiently given a contracted range. Specifically, computing $g(i,j,c)$ such that $A[i]=A[j]=c$ takes $O(1)$ time. This computation may use an additional $O(n)$-space static data structure.
\end{enumerate}

The \textbf{Range Maximum $g$-value Color} problem is to preprocess $A$ and the function $g$ to efficiently answer queries of the following form:
\begin{quote}
    Given two indices $i, j$ of $A$, determine: $\arg\max_{c \in A[i \dots j]} \{ g(i, j, c) \}$
\end{quote}
In other words, determine any single color $c$ corresponding to the maximum value of $g(i,j, c)$.

\subsection{Generalization to trees}

This problem generalizes naturally to trees. Instead of an array, we are given a \textbf{rooted}, node-colored tree $\mathcal{T}$ and a function $g$ defined on tree paths. An input to $g$ consists of a triple $(u, v, c)$, where $u, v$ are nodes in $\mathcal{T}$ and $c$ is a color occurring on the simple path $P(u, v)$. From now on, we will consider $\mathcal{T}$ to be rooted at node $1$.

The properties of $g$ remain analogous. Crucially, the \textbf{Range Contraction} property (Property 3) translates as follows:
\[
    g(u, v, c) = g(u', v', c)
\]
where $u'$ is the node on $P(u, v)$ closest to $u$ such that $\mathcal{T}[u'] = c$, and $v'$ is the node on $P(u, v)$ closest to $v$ such that $\mathcal{T}[v'] = c$, where $\mathcal{T}[u]$ denotes the color of node $u$ in the tree $\mathcal{T}$.

We refer to this generalized problem as the \textbf{Path Maximum $g$-value Color} problem. This is the main problem we focus on in our work. In the following sections, we present a linear-space data structure supporting path maximum $g$-value color queries in $O(\sqrt{n/w})$ time, requiring $O(n \sqrt{nw})$ preprocessing time.

\subsection{Applications of the path maximum g-value color problem} 

To motivate why this abstract definition is interesting in its own right, we introduce the \textbf{Range Maximum Weight Color} problem, which falls under the broader umbrella of \textbf{weighted range queries}. Given a color array $A$ and a weight array $W$ (representing node weights), we wish to find the color in a range $A[i \dots j]$ that maximizes the sum of its associated weights. Formally:
\begin{align*}
    \arg\max_{c \in A[i \dots j]} \left\{ \sum_{k \in [i, j] : A[k]=c} W[k] \right\}    
\end{align*}

This is a specific instance of the maximum $g$-value color problem, as the weight summation function satisfies all properties of $g$ (specifically, the sum of weights for color $c$ in range $[i, j]$ is identical to the sum in the contracted range $[i', j']$). 
 
This problem naturally generalizes to trees, introducing the \textbf{Path Maximum Weight Color} problem as a fundamental operation for \textbf{weighted path queries}. By using the \textbf{virtual tree} data structure, and Lemma~\ref{lemma:constant_time_lca_query} presented in Section~\ref{section:prerequisites}, we can easily verify that:
\begin{align}
\label{relation:g_path_colored_sum}
    g(u_c, v_c,c)=
    \begin{cases}
         g(u_c, r_c, c)+g(v_c,r_c,c)-2\cdot g(lca,r_c,c) & lca\notin P(u_c,v_c)\\
         g(u_c, r_c, c)+g(v_c,r_c,c)-2\cdot g(lca,r_c,c)+W[lca] & lca \in P(u_c,v_c)
    \end{cases}
\end{align}

where $lca$ is the lowest common colored ancestor of $u_c$ and $v_c$ (which requires $\mathcal{T}$ to be rooted), $u_c,v_c$ are both $c$-colored nodes, and $r_c$ is the root of the virtual tree of color $c$. We can verify that $g(u_c,v_c,c)$ can be computed in $O(1)$ time using relation~\ref{relation:g_path_colored_sum} if we store for each $u_c$ the values $g(u_c,r_c,c)$ in a table of size $O(n)$.

Both on arrays and on trees, the \textbf{Maximum Weight Color} problem acts as a generalization for standard frequency queries, effectively uniting them as \textbf{weighted frequency queries}. Specifically, it encompasses both the \textbf{Mode} problem (where $W[k]=1$ for all $k$) and the \textbf{Least Frequent Element} problem (where $W[k]=-1$ for all $k$, allowing us to maximize the negated sum).

\subsubsection{Practical Example: Financial Time Series.}
As an additional motivation towards the independent study of the \textbf{Range Maximum $g$-value Color} problem, consider a log of chronologically ordered stock records: ``Stock $c$ changed by value $v$''. Storing stock IDs in $A$ and changes as weights in $W$, the Maximum Weight Color problem answers: \textit{``Which stock had the highest net increase/decrease between time $i$ and $j$?''}

Furthermore, our general $g$-value framework allows for more complex queries, such as identifying the stock with the \textbf{maximum mean fluctuation} per record within a time period, provided the mean is calculated over the stock's records. This demonstrates that the $g$-value framework extends beyond standard frequency queries to elegantly handle cumulative weights and penalties.

\subsection{Hierarchy and Contribution}
We conclude this section by outlining the difficulty hierarchy. We focus on solving the most general variant: the Path Maximum $g$-value Color problem on trees. We present a linear-space data structure for this problem with a query time of $O(\sqrt{n/w})$. Remarkably, this matches the query time of the fastest known data structure for the range mode query problem (Chan et al.~\cite{chan2014linear}), which corresponds to the simplest case in our hierarchy.

\begin{figure}[H]
    \centering
    \includegraphics[width=0.4\linewidth]{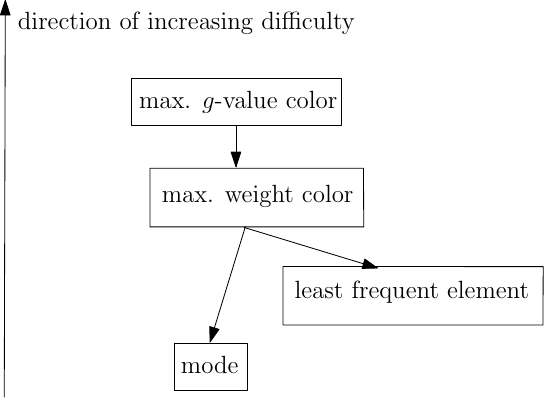}
    \caption{The hierarchy of problems ordered by increasing difficulty and generality. The \textbf{max.\ $g$-value color} problem generalizes the \textbf{max.\ weight color} problem, which covers both the \textbf{mode} and \textbf{least frequent element} problems.}
    \label{fig:problem_hierarchy}
\end{figure}

\section{Prerequisites}
\label{section:prerequisites}
We mention that we consider the colored tree $\mathcal{T}$ given as input to be rooted at node $1$. Prior to presenting our solution, we introduce several prerequisite results and definitions. We rely on the following standard results for tree data structures:

\begin{lemma}[Bender and Farach-Colton~\cite{bender2004level}; Berkman et al.~\cite{berkman1994finding}; Dietz~\cite{dietz1991finding}]
    \label{lemma:constant_time_level_ancestor_queries}
    There exists a linear-space data structure supporting level ancestor queries on trees in $O(1)$ time, requiring $O(n)$ preprocessing time.
\end{lemma}

\begin{lemma}[Bender and Farach-Colton~\cite{bender2000lca}; Berkman and Vishkin~\cite{berkman1993recursive}]
    \label{lemma:constant_time_lca_query}
    There exists a linear-space data structure supporting lowest common ancestor (LCA) queries on trees in $O(1)$ time, requiring $O(n)$ preprocessing time.
\end{lemma}

\subsection{Virtual Trees}
\label{subsection:prerequisites:subsubsection:intermediary_results_on_colored_trees}

We define \textbf{virtual trees}, which are essential for our data structure.

\begin{definition}
    Given a colored tree $\mathcal{T}$ and a color $c$, the \textbf{virtual tree of color $c$}, denoted $\mathcal{V}_c$, is a rooted tree defined over the set of nodes $\{v \in \mathcal{T} \mid \text{color}(v) = c\}$.
    
    The structure of $\mathcal{V}_c$ is defined by the following parent relationship: for any node $v \in \mathcal{V}_c$, the parent of $v$ in $\mathcal{V}_c$ is the nearest proper ancestor of $v$ in $\mathcal{T}$ that also has color $c$.
\end{definition}

To ensure $\mathcal{V}_c$ is a single connected tree (rather than a forest), we introduce a virtual root node $r_c$ for each color $c$. If a node $v$ of color $c$ has no ancestors of color $c$ in $\mathcal{T}$, we set its parent in $\mathcal{V}_c$ to be $r_c$.

It is possible to construct all virtual trees $\mathcal{V}_c$ for all colors $c$ in $O(n)$ total time using a single depth-first traversal of $\mathcal{T}$. Since the node sets for each color form a partition of the colored nodes in $\mathcal{T}$, the total space required is linear, $O(n)$. We defer the algorithm for computing all the virtual trees $\mathcal{V}_c$ to appendix~\ref{appendix:virtual_trees}. 

By applying standard LCA and Level Ancestor structures (Lemmas~\ref{lemma:constant_time_level_ancestor_queries} and \ref{lemma:constant_time_lca_query}) to these virtual trees, we obtain the following corollary:

\begin{corollary}
    \label{corollary:virtual_tree_queries}
    There exists a linear-space data structure, constructed in $O(n)$ time, that answers the following queries in $O(1)$ time:
    \begin{enumerate}
        \item Given two nodes $u, v \in \mathcal{T}$ of the same color $c$, determine the lowest common ancestor of $u$ and $v$ \textbf{with respect to color} $c$ (i.e., the LCA of $u$ and $v$ in $\mathcal{V}_c$).
        \item Given a node $u$ of color $c$ and an integer $k$, determine the $k$-th ancestor of $u$ \textbf{of color} $c$ (i.e., the $k$-th ancestor of $u$ in $\mathcal{V}_c$).
    \end{enumerate}
\end{corollary}
\begin{proof}
    For each distinct color $c$ present in $\mathcal{T}$, we explicitly construct the virtual tree $\mathcal{V}_{c}$. Over each $\mathcal{V}_{c}$, we build instances of the data structures described in Lemmas~\ref{lemma:constant_time_level_ancestor_queries} and~\ref{lemma:constant_time_lca_query}.
    
    Let $n_c = |V(\mathcal{V}_c)|$. The construction and preprocessing for a specific color $c$ takes $O(n_c)$ time and space. Summing over all colors, the total preprocessing time is $\sum_{c} O(n_c) = O(n)$, and the total space is similarly $\sum_{c} O(n_c) = O(n)$. The queries are answered by querying the specific structure for $\mathcal{V}_c$ in $O(1)$ time.
\end{proof}

\section{Blocking Technique}
\label{section:blocking_technique}

Krizanc et al.~\cite{krizanc2005range}, Chan et al.~\cite{chan2014linear, chan2012linear}, and Durocher et al.~\cite{durocher2016linear} successfully employed blocking techniques to solve non-trivial frequency queries on arrays and trees. We adapt their framework to develop a multi-level blocking strategy, allowing us to achieve the desired query time for path maximum $g$-value color queries on trees.

\begin{lemma}[Durocher et al.~\cite{durocher2016linear}]
\label{lemma:durocher_blocking}
Let $\mathcal{T}$ be a tree with $n$ nodes, and let $t < n$ be an integer parameter (the \textit{blocking factor}). There exists a subset of marked nodes $M_t \subseteq V(\mathcal{T})$ and an associated rooted tree $\mathcal{T}_t$ with vertex set $M_t$, satisfying the following properties:
\begin{enumerate}
    \item \textbf{Size:} $|M_t| = O(n/t)$;
    \item \textbf{LCA Closure:} $M_t$ is closed under the lowest common ancestor operation (i.e., for any $u, v \in M_t$, the lowest common ancestor of $u$ and $v$ in $\mathcal{T}$ is also in $M_t$);
    \item \textbf{Gap Size:} Any simple path in $\mathcal{T}$ consisting entirely of unmarked nodes has length at most $t$.
\end{enumerate}
Two nodes $u, v \in M_t$ are adjacent in $\mathcal{T}_t$ if and only if the simple path between them in $\mathcal{T}$ contains no other nodes from $M_t$. The set $M_t$ and the tree $\mathcal{T}_t$ can be constructed in $O(n \log n)$ time.
\end{lemma}

\begin{lemma}[Durocher et al.~\cite{durocher2016linear}]
\label{lemma:durocher_real_blocking}
    Given a tree $\mathcal{T}$ with $n$ nodes, an integer $t < n$, and a set of marked nodes $M_t$ (obtained according to Lemma~\ref{lemma:durocher_blocking}), there exists a partition of the vertex set $V(\mathcal{T})$, denoted by $\mathcal{B}_t = \{B_1, \dots, B_k\}$, satisfying the following properties:
    \begin{enumerate}
        \item \textbf{Block Count:} The number of blocks is $k = O(n/t)$;
        \item \textbf{Block Size:} For each block $B \in \mathcal{B}_t$, the number of nodes is $|B| = O(t)$;
        \item \textbf{Connectivity:} For each block $B \in \mathcal{B}_t$, the subgraph of $\mathcal{T}$ induced by the nodes in $B$ is connected.
    \end{enumerate}
    The partition $\mathcal{B}_t$ can be computed in $O(n \log n)$ time.
\end{lemma}

We defer the full details of the proofs of Lemmas~\ref{lemma:durocher_blocking} and~\ref{lemma:durocher_real_blocking} to Appendix~\ref{appendix:blocking_results_proofs}. Further, for our approach, we require a multilevel partition of the nodes of $\mathcal{T}$. From Lemma~\ref{lemma:durocher_real_blocking}, we derive the following corollary, which formalizes the construction of a hierarchical partition.

\begin{corollary}
\label{corollary:multi_level_blocking}
Given a colored tree $\mathcal{T}$, two blocking factors $t_1 < t_2$, and the set of marked nodes $M_{t_1}$ with its associated tree $\mathcal{T}_{t_1}$ (obtained according to Lemma~\ref{lemma:durocher_blocking}), there exists a subset $M_{t_2} \subseteq M_{t_1}$ and a partition $\mathcal{B}_{t_2/t_1}$ of $M_{t_1}$, and a tree $\mathcal{T}_{t_2}$.

These can be computed in $O(|M_{t_1}| \log |M_{t_1}|)$ time and satisfy the following properties:
\begin{enumerate}
    \item \textbf{Consistency:} The set $M_{t_2}$ is a valid set of marked nodes for $\mathcal{T}$ with blocking factor $t_2$ (satisfying all properties of Lemma~\ref{lemma:durocher_blocking}).
    \item \textbf{Relative Gap:} Any simple path in the compressed tree $\mathcal{T}_{t_1}$ consisting entirely of nodes from $M_{t_1} \setminus M_{t_2}$ has length $O(t_2/t_1)$.
    \item \textbf{Partition Structure:} The partition $\mathcal{B}_{t_2/t_1}$ divides the nodes of $M_{t_1}$ into $O(|M_{t_1}| \cdot \frac{t_1}{t_2})$ blocks, each of size $O(t_2/t_1)$, such that the subgraph of $\mathcal{T}_{t_1}$ induced by each block is connected.
\end{enumerate}
\end{corollary}
\begin{proof}
    We apply the construction procedures from Lemma~\ref{lemma:durocher_blocking} and Lemma~\ref{lemma:durocher_real_blocking} directly to the compressed tree $\mathcal{T}_{t_1}$, treating it as the input tree with a blocking factor of $t' = t_2/t_1$.
    Since $|V(\mathcal{T}_{t_1})| = |M_{t_1}|$, the construction requires $O(|M_{t_1}| \log |M_{t_1}|)$ time. The resulting marked nodes $M_{t_2}$ and partition $\mathcal{B}_{t_2/t_1}$ satisfy the size and connectivity properties with respect to $\mathcal{T}_{t_1}$ by definition. Furthermore, since $M_{t_1}$ preserves the LCA closure property of $\mathcal{T}$, and $M_{t_2}$ is constructed to preserve LCA closure on $\mathcal{T}_{t_1}$, $M_{t_2}$ maintains the necessary topological properties for $\mathcal{T}$.
\end{proof}

Crucially, the partition $\mathcal{B}_{t_2/t_1}$ of $M_{t_1}$ implicitly defines a coarser partition $\mathcal{B}_{t_2}$ of the original vertex set $V(\mathcal{T})$. A high-level block $B \in \mathcal{B}_{t_2}$ consists of a set of connected low-level blocks from $\mathcal{B}_{t_1}$. To enumerate the nodes of a high-level block $B \in \mathcal{B}_{t_2}$, we identify the corresponding block of marked nodes $B' \in \mathcal{B}_{t_2/t_1}$ (where $B' \subseteq M_{t_1}$) and iterate through the low-level block $B''_u \in \mathcal{B}_{t_1}$ associated with each $u \in B'$.

\begin{corollary}
    \label{corollary:block_trees}
    Given a colored tree $\mathcal{T}$, two blocking factors $t_1<t_2$, and the respective partitions into blocks computed according to Lemma~\ref{lemma:durocher_real_blocking} and Corollary~\ref{corollary:multi_level_blocking}, we can compute in $O(|M_{t_1}|)$ time, for each block $B_2\in \mathcal{B}_{t_2/t_1}$, a tree $\mathcal{T}_{t_2/t_1}[B_2]$, which is the subtree of $T_{t_1}$ induced by the nodes in $B_{2}$ (which are also present in $M_{t_1}$).
\end{corollary}
\begin{proof}
    By Property 3 of Corollary~\ref{corollary:multi_level_blocking}, the subgraph of $\mathcal{T}_{t_1}$ induced by any block $B \in \mathcal{B}_{t_2/t_1}$ is connected. Since $\mathcal{T}_{t_1}$ is a tree, any connected subgraph is also a tree.
    To construct these explicit tree representations, we iterate through the edges of $\mathcal{T}_{t_1}$. For each edge $(u, v) \in E(\mathcal{T}_{t_1})$, if both $u$ and $v$ belong to the same block $B \in \mathcal{B}_{t_2/t_1}$, we add the edge to $\mathcal{T}_{t_2/t_1}[B]$. This traversal visits every edge exactly once, resulting in a total runtime of $O(|V(\mathcal{T}_{t_1})|) = O(|M_{t_1}|)$.
\end{proof}
We refer to $\mathcal{T}_{t_2/t_1}[B_2]$ as the \textbf{block tree} of $B_2$. We depict an example of a $2$-level block partition in Figure~\ref{fig:block_hierarchy}.

\begin{figure}[H]
    \centering
    \includegraphics[width=0.3\linewidth]{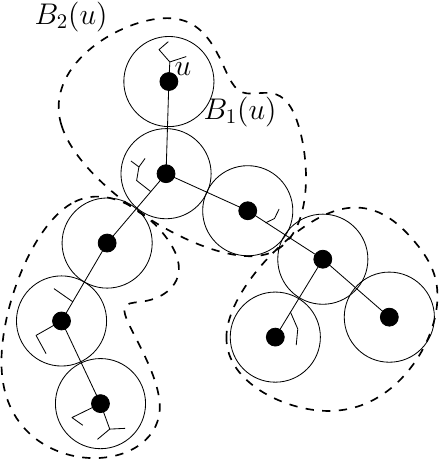}
    \caption{A $2$-level partition of a tree. The representatives of a block are depicted with black dots. $t_1$-blocks are depicted by circles, and $t_2$-blocks by dashed splinoids. An example representative node $u$ is chosen, and its respective blocks $B_1(u)$ and $B_2(u)$ are depicted.}
    \label{fig:block_hierarchy}
\end{figure}

We define the notions of home blocks for nodes and node representatives for blocks.
\begin{definition}[Home Blocks and Representatives]
    For any node $u \in V(\mathcal{T})$ and level $k$, let $B_k(u)$ denote the unique block in the partition $\mathcal{B}_{t_k}$ that contains $u$.
    We assume each block $B \in \mathcal{B}_{t_k}$ is associated with a unique representative node in $M_{t_k}$.
\end{definition}

This definition allows us to examine the $3$-level nested structure around a node $u$. The node $u$ resides in a local block $B_1(u)$, which is nested within $B_2(u)$, which is further nested within $B_3(u)$. 

For any level $k\in \{1,2,3\}$, we refer to any block $B\in \mathcal{B}_{t_k}$ as a $k$-block. We call two $k$-block representatives $u$ and $v$ \textit{adjacent} if there is no other $k$-block representative $z$ appearing on the path $P(u,v)$. 

Similarly to how every $k$-block is a subtree of $\mathcal{T}$, the subgraph connecting adjacent $k$-block representatives is also a subtree. We formalize this statement via the following lemma. 

\begin{lemma}
    \label{lemma:adjacent-representatives-connection}
    For any adjacent $k$-block representatives $u,v$, the set of nodes reachable from any node $x\in P(u,v) \setminus \{u,v\}$, via paths that do not intersect $u$ or $v$, induces a subtree of the original tree $\mathcal{T}$. We denote this subtree by $\mathcal{T}(u,v)$.

    Additionally, the only $k$-block representatives in $\mathcal{T}(u,v)$ are $u$ and $v$ (which bound the subtree).
\end{lemma}
\begin{proof}
    \begin{figure}[H]
    \centering
    \begin{tabular}{cc}
        \includegraphics[width=0.6\linewidth]{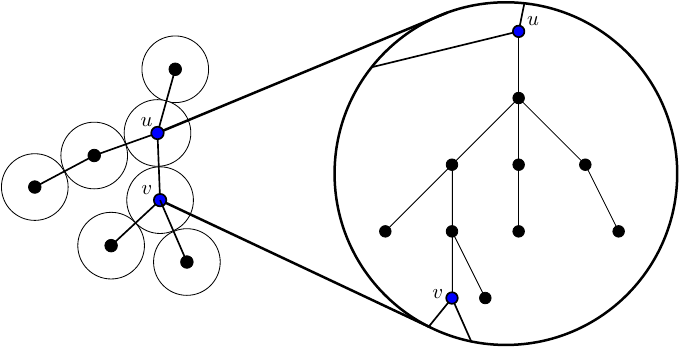} &
        \includegraphics[width=0.33\linewidth]{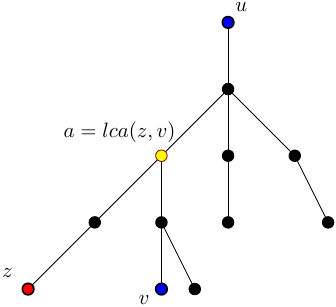} \\
         (a) Valid subtree $\mathcal{T}(u, v)$ & (b) Invalid configuration
    \end{tabular}
    \caption{Depiction of a valid subtree $\mathcal{T}(u,v)$, and an invalid configuration, connecting the $k$-block representatives $u,v$. In (a), $\mathcal{T}(u,v)$ is a valid tree connecting the representatives $u$, $v$. In (b), for the sake of contradiction, we introduce a new representative $z \in \mathcal{T}(u,v)$. The node $a=lca(z,v)$ must merge with the path $P(u,v)$, but according to the LCA property of blocking, $a$ is also a representative, which yields a contradiction.}
    \label{fig:color_endpoints_small_blocks_cases}
    \end{figure}
    
    The first part of the lemma follows directly from the definition of trees; any connected component of a tree obtained by removing specific edges or nodes remains a tree.
    
    To prove the second part, note that by the LCA-closure property of our blocking structure (Lemma~\ref{lemma:durocher_blocking}), the lowest common ancestor of any two representatives is also a representative. Because $u$ and $v$ are adjacent representatives, their LCA must lie on the path $P(u,v)$. Since no other representatives exist on this path, $lca(u,v)$ must be exactly $u$ or $v$. Without loss of generality, let $u$ be the ancestor of $v$.

    Assume, for the sake of contradiction, that some other $k$-block representative $z$ is part of $\mathcal{T}(u,v)$. Because $z$ is in $\mathcal{T}(u,v)$, the path from $z$ to $v$ must join $P(u,v)$ at some node $x$ that is strictly between $u$ and $v$ (i.e., $x \in P(u,v) \setminus \{u,v\}$). Consequently, the node $a = lca(z,v)$ is exactly this attachment node $x$. According to the LCA property, $a$ must be a $k$-block representative. However, this implies that $a$ is a representative located on $P(u,v)$ strictly between $u$ and $v$, which contradicts the assumption that $u$ and $v$ are adjacent.
\end{proof}

\subsection{Parameter Selection}
For our specific data structure, we employ a $3$-level partition of $\mathcal{T}$. We assume the first blocking factor $t_1$ is provided as a parameter. We then define the higher-level factors as follows:
\[
    t_2 = t_1 \log \log n, \quad t_3 = t_2 \sqrt{\log n}
\]
During the preprocessing phase, we compute the associated structures $M_{t_i}$, $\mathcal{T}_{t_i}$, and $\mathcal{B}_{t_i}$ for $1 \leq i \leq 3$. Specifically, we construct the full block partitions $\mathcal{B}_{t_i}$ using the hierarchical method described in Corollary~\ref{corollary:multi_level_blocking}. Based on the established results, this entire preprocessing step requires $O(n \log n)$ time.

\section{Overview and Strategy}
\label{section:important_regions_subtasks}

Having established the prerequisites, we now provide a high-level overview of our query answering strategy. Our approach refines the blocking technique of Durocher et al.~\cite{durocher2016linear}. In their solution, the query algorithm identifies a set of $O(\sqrt{n/w})$ candidate nodes. For every distinct color appearing in this set of nodes, they perform a frequency verification step taking $O(\log\log n)$ time.

To improve upon this, we introduce a hierarchical candidate selection process. Instead of a single set, we identify two distinct sets of \textbf{candidate nodes}, denoted $S_1$ and $S_2$. These sets act as proxies for the colors we need to verify:
\begin{itemize}
    \item \textbf{Primary Candidate Nodes ($S_1$):} A small set of nodes ($|S_1|=O(t_1)$). For every color $c$ appearing in $S_1$, we perform a standard verification taking $O(\log\log n)$ time.
    \item \textbf{Secondary Candidate Nodes ($S_2$):} A larger set of nodes ($|S_2| = O(t_2)$). For every color $c$ appearing in $S_2$, we store a triple $(c,l,r)$, where $l$ and $r$ are the first occurrences of $c$, when traversing the path in the direction $i\rightarrow j$ and $j\rightarrow i$, respectively. We need only $O(1)$ time to compute $g(c,l,r)$ for one such color.
\end{itemize}
This split allows us to process a larger total number of candidate nodes (and thus potential colors) without degrading the overall query complexity. 

\subsection{Decomposition into Disjoint Subtasks}
\label{subsection:decomposition_into_subtasks}
To efficiently identify these candidate nodes, we decompose the query based on the hierarchy of blocks defined in Section~\ref{section:blocking_technique}. We decompose the problem into disjoint subtasks by partitioning the set of distinct colors present on the query path. Each subtask identifies a unique subset of colors based on their presence or absence in the hierarchical blocks surrounding $i$ and $j$.

Let $C_{path} = C(P(i, j))$ denote the set of all distinct colors present on the query path.
For any block $B$, let $C(B)$ denote the set of distinct colors present in that block.

We define the complement of a block's color set with respect to the \textbf{entire tree}. Let $\mathcal{U}$ be the universe of all distinct colors in the tree $\mathcal{T}$. We define:
\[
    \overline{C(B)} = \mathcal{U} \setminus C(B)
\]
This represents the set of all colors in the tree that are \textit{not} present in block $B$.

The set of path colors $C_{path}$ is partitioned into the following disjoint subsets $C_1, \dots, C_{10}$ using intersection operations on these color sets:

\begin{itemize}
    \item \textbf{1. The Global Colors:}
    $
       C_1 = C_{path} \cap \overline{C(B_3(i))} \cap \overline{C(B_3(j))}
    $.
    
    Colors on the path that are not present in the Level 3 blocks of either endpoint.

    \item \textbf{2. The Level 3 Ascent:}
    $
       C_2 = C_{path} \cap C(B_3(i)) \cap \overline{C(B_2(i))} \cap \overline{C(B_3(j))}
    $.
    
    Colors on the path that are in $i$'s Level 3 block, but not in $i$'s Level 2 block or $j$'s Level 3 block.

    \item \textbf{3. The Level 2 Ascent:}
    $
       C_3 = C_{path} \cap C(B_2(i)) \cap \overline{C(B_1(i))} \cap \overline{C(B_3(j))}
    $.
    
    Colors on the path that are in $i$'s Level 2 block, but not in $i$'s Level 1 block or $j$'s Level 3 block.

    \item \textbf{4. The Level 3 Descent:}
    $
       C_4 = C_{path} \cap C(B_3(j)) \cap \overline{C(B_2(j))} \cap \overline{C(B_3(i))}
    $.
    
    Symmetric to $C_2$.

    \item \textbf{5. Level 3 Intersection:}
    $
       C_5 = C_{path} \cap C(B_3(i)) \cap C(B_3(j)) \cap \overline{C(B_2(i))} \cap \overline{C(B_2(j))}
    $.
    
    Colors on the path present in both Level 3 blocks, but in neither Level 2 block.

    \item \textbf{6. Mixed Intersection:}
    $
       C_6 = C_{path} \cap C(B_2(i)) \cap C(B_3(j)) \cap \overline{C(B_1(i))} \cap \overline{C(B_2(j))}
    $.
    
    Colors on the path present in $i$'s Level 2 block and $j$'s Level 3 block, excluding inner blocks.

    \item \textbf{7. The Level 2 Descent:}
    $
       C_7 = C_{path} \cap C(B_2(j)) \cap \overline{C(B_1(j))} \cap \overline{C(B_3(i))}
    $.
    
    Symmetric to $C_3$.

    \item \textbf{8. Mixed Intersection:}
    $
       C_8 = C_{path} \cap C(B_3(i)) \cap C(B_2(j)) \cap \overline{C(B_2(i))} \cap \overline{C(B_1(j))}
    $.
    
    Symmetric to $C_6$.

    \item \textbf{9. Level 2 Intersection:}
    $
       C_9 = C_{path} \cap C(B_2(i)) \cap C(B_2(j)) \cap \overline{C(B_1(i))} \cap \overline{C(B_1(j))}
    $.
    
    Colors on the path present in both Level 2 blocks, excluding both Level 1 blocks.

    \item \textbf{10. Local Colors:}
    $
       C_{10} = C_{path} \cap \Big( C(B_1(i)) \cup C(B_1(j)) \Big)
    $.
    
    Colors on the path present in the smallest (Level 1) blocks of either endpoint.
\end{itemize}

In the subsequent analysis, we will omit the subtasks marked as symmetric to others (e.g., $C_4$, $C_7$, $C_8$), as their solutions are algorithmically identical to their counterparts, differing only in orientation. 

For each of the subtasks 4-9, we will identify sets of candidate nodes of size $O(t_2)$, together with the corresponding endpoints $l,r$, and add them to the set $S_2$. For subtasks $1-3$ and $10$, we will identify sets of candidates of size $O(t_1)$ or $O(1)$, and add them to $S_1$. At the end, we process the sets $S_2$ and $S_1$, in time $O(|S_2|)=O(t_2)$ and $O(|S_1|\cdot\log\log n)=O(t_1\cdot \log\log n)$, respectively, and return the color with the largest $g$-value. Since $t_1\log\log n=t_2$, processing $S_1$ and $S_2$ takes $O(t_2)$ in total.

\subsubsection{Proof of Completeness}
\label{subsection:proof_of_completeness}
We prove that these subtasks partition $C_{path}$. Consider any color $c \in C_{path}$. Based on the block hierarchy, we can determine the membership of $c$ in the color sets of the blocks surrounding $i$ and $j$.
Let $P_u$ be the proposition $c \in C(B_3(u))$ and $Q_u$ be the proposition $c \in C(B_2(u))$ for $u \in \{i, j\}$. Since $B_2(u) \subseteq B_3(u)$, we have $Q_u \implies P_u$.
The definitions of $C_1$ through $C_9$ correspond precisely to the disjoint logical conjunctions of these propositions (e.g., $C_2$ corresponds to $P_i \land \neg Q_i \land \neg P_j$). The set $C_{10}$ captures the remaining cases where the color is present in $C(B_1(i))$ or $C(B_1(j))$.
Thus, every color $c \in C_{path}$ satisfies exactly one of these conditions, ensuring that $\bigcup_{k=1}^{10} C_k = C_{path}$ and all $C_k$ are disjoint.

\section{Solving the Subtasks}
\label{section:solving_subtasks}

In this section, we present the algorithmic solutions for the subtasks defined previously. While the query procedures are concise, the precomputation strategies vary in complexity across subtasks; we sketch the data structures here, deferring full construction details to Appendix~\ref{appendix:precomputation_algorithms}. We first state a key result from Durocher et al.~\cite{durocher2016linear} that enables efficient verification for our primary candidate set $S_1$.

\begin{lemma}[Durocher et al.~\cite{durocher2016linear}]
\label{lemma:durocher:path_frequency_query}
There exists an $O(n)$-space data structure that supports lowest colored ancestor queries on trees in $O(\log\log n)$ time.
\end{lemma}

To verify any $c\in S_1$, we must evaluate $g(i, j, c)$, which requires identifying the first and last occurrences of $c$ on $P(i, j)$. Using the lowest colored ancestor structure from Lemma~\ref{lemma:durocher:path_frequency_query}, we locate these endpoints in $O(\log \log n)$ time, and compute $g$ in $O(1)$. We analyze the space complexity, query time, and preprocessing time for each subtask. Note that any candidate color identified in these subtasks is added to the candidate set $S_1$ (verification taking $O(\log \log n)$ time) or $S_2$ (verification taking $O(1)$ time), as per the strategy in Section~\ref{section:important_regions_subtasks}.

\subsection{Solving Subtask 1: The Global Region}
\label{section:generalized_Z_query_trees:subsection:solvingsubtasks:subtask1}

We precompute a table $T$ of size $O\left( (n/t_3)^2\right)$, indexed by the identifiers of the blocks in $\mathcal{B}_{t_3}$. For every pair of blocks $B_u, B_v \in \mathcal{B}_{t_3}$ with representatives $u, v$, the entry $T[u][v]$ stores the color $c$ that maximizes $g(u, v, c)$, restricted to colors appearing on the path $P(u, v)$ but not in $C(B_u) \cup C(B_v)$. During a query, we identify the blocks $u=B_3(i)$ and $v=B_3(j)$, retrieve the candidate $c=T[u][v]$ in $O(1)$ time, and add it to the primary candidate set $S_1$. 

\textbf{Complexity:} Preprocessing $O(n \cdot (n/t_3))$; Space $O((n/t_3)^2)$ words; Query Time $O(1)$.

\subsection{Solving Subtask 2: The Level 3 Ascent}
\label{section:generalized_Z_query_trees:subsection:solvingsubtasks:subtask2}

We employ a lookup table $T$ of size $O(\frac{n}{t_2} \cdot \frac{n}{t_3})$, indexed by pairs $(u, v)$ representing a Level 2 block and a Level 3 block. The entry $T[u][v]$ stores the relative index (using $O(\log \log n)$ bits) of a specific Level 1 block (or $t_1$-node) inside $B_3(i)$ that contains at least one occurrence of the optimal color. During a query, we iterate through every distinct color $c'$ in this retrieved $t_1$-node's subtree and add each to $S_1$.

\textbf{Complexity:} Preprocessing $O(n \cdot (n/t_3))$; Space $O((n/t_3)(n/t_2) \log \log n)$ bits; Query Time $O(t_1)$.

\subsection{Solving Subtask 3: The Level 2 Ascent}
\label{section:generalized_Z_query_trees:subsection:solvingsubtasks:subtask3}

This solution is analogous to Subtask 2. We use a table $T$ of size $O(\frac{n}{t_1} \cdot \frac{n}{t_3})$. For a pair of blocks $B_u$, $B_v$ (Level 1 and Level 3), $T[u][v]$ stores the index of a $t_1$-node inside $B_2(u)$ containing the optimal color. All colors in that $t_1$-node are added to $S_1$.

\textbf{Complexity:} Preprocessing $O(n \cdot (n/t_3))$; Space $O((n/t_3)(n/t_1) \log \log n)$ bits; Query Time $O(t_1)$.

\subsection{Prerequisites for Subtasks 5, 6 and 9}
\label{subsection:Prerequisites_subtask9_6}

To efficiently solve the intersection subtasks, we require advanced data structures that allow us to locate specific color occurrences within the block hierarchy in constant time. We refer to trees with size bounded by $s = \frac{\log n}{8 \log \log n}$ as \textbf{small trees}.

\begin{lemma}
    \label{lemma:small_tree_search}
    There exists an $O(n)$-space data structure, constructible in $O(n)$ time, that answers the following query in $O(1)$ time: Given a small tree $\tau$, a binary string $\sigma_{V}$ of length $|V(\tau)|=O(\frac{\log n}{\log\log n})$, and a node $u \in V(\tau)$, return the lowest ancestor of $u$ marked with $1$ in $\sigma_V$. If no such ancestor exists, return a marked node $v'$ such that no ancestor of $v'$ (other than potentially $u$) is marked.
\end{lemma}
\begin{proof}
    We employ the tabulation method. We precompute answers for all possible canonical instances of small trees. The number of distinct rooted trees with size up to $s$ is bounded by $O(2^{s\log s})$. For each tree topology $\tau$, there are $2^{|V(\tau)|} \le 2^s$ possible bitstrings $\sigma_V$. Thus, the total number of distinct query inputs $(\tau, \sigma, u)$ is bounded by:
    \[
        O(2^{s\log s} \cdot 2^s \cdot s) = O(2^{2\log\log n\frac{\log n}{8 \log \log n}} \cdot \log n) = O(n^{1/4} \log n) \ll O(n).
    \]
    We construct a lookup table $T$ indexed by the canonical ID of $\tau$, the bitstring $\sigma_V$, and the node index $u$. We populate $T$ using a standard DFS traversal for each instance. During the query phase, we map the input small tree to its canonical ID and retrieve the answer from $T$ in $O(1)$ time. The total space and preprocessing time are sublinear.
\end{proof}

We apply the same logic to edge markings rather than node markings, and summarize this result via the following corollary.

\begin{corollary}
    \label{corollary:smalle_tree_edge_search}
    There exists an $O(n)$-space data structure, constructible in $O(n)$ time, that answers the following query in $O(1)$ time: Given a small tree $\tau$, a binary string $\sigma_{E}$ of length $|E(\tau)|=O(\frac{\log n}{\log\log n})$, and an edge $e \in E(\tau)$, return the lowest edge above $e$ marked with $1$ in $\sigma_E$. If no such edge exists, return a marked edge $e'$ such that no edge above $e'$ (other than potentially $e$) is marked.
\end{corollary}
\begin{proof}
    For each edge $e$, let $low(e)$ denote its lowest endpoint (the node furthest from the root). We build the node marking $\sigma_V$ the following way:
    \begin{align*}
        \sigma_{V}(u)=
        \begin{cases}
            \sigma_E(e), & \text{if } low(e)=u\\
            0, & \text{if } u\text{ is the root of $\tau$}
        \end{cases}
    \end{align*}
    We then build the data structure $DS_V$ from Lemma~\ref{lemma:small_tree_search} for the pair $(\tau, \sigma_V)$. In order to answer a query $e$ on $(\tau, \sigma_E)$, we determine $u \gets low(e)$, and return $DS_V(u)$.
\end{proof}

We combine the node marking and edge marking small tree search data structures to obtain the following result related to finding occurrences of a color $c$ within a block hierarchy.

\begin{lemma}
    \label{lemma:hashing_multi_level_color_search}
    Given a colored tree $\mathcal{T}$ and blocking factors $t_1 < t_2$ with $t_2/t_1 = O(\log n)$, there exists a linear-space perfect hashing structure (preprocessed in $O(n)$ expected time) that, given a color $c$ and a $t_2$-block $B$, returns in $O(1)$ time:
    \begin{enumerate}
        \item a binary string $\sigma_V$ of length $t_2/t_1$ indicating which $t_1$-blocks inside $B$ contain $c$.
        \item a binary string $\sigma_E$ of length $t_2/t_1$ indicating which of the paths between adjacent $1$-block representatives contain the color $c$.
    \end{enumerate}
\end{lemma}
\begin{proof}
    We construct a dictionary of valid pairs. Let $S$ be the set of triples $(c, B, \sigma_V)$, where $c$ is a color present in the $t_2$-block $B$, and $\sigma_V$ is the bitmask representing the presence of $c$ in the constituent $t_1$-blocks. The $t_1$ nodes inside the big $t_2$ block form a subtree $\tau$ of size $O(\log n)$. The size of the dictionary $S$ is bounded by $O(n)$ because each node in $\mathcal{T}$ belongs to exactly one $t_1$-block, contributing exactly one bit to one mask in the entire structure.
    
    We store $S$ using a static perfect hashing scheme (e.g., Fredman, Komlós and Szemerédi~\cite{10.1145/828.1884}, or Belazzougui et al.~\cite{belazzougui2009hash}, Hagerup and Tholey~\cite{hagerup2001efficient}), allowing $O(1)$ worst-case lookups to retrieve $\sigma_V$ for a pair $(c, B)$, or returning null if $c \notin C(B)$. By a similar perfect hashing approach, we build a data structure to retrieve the edge marking $\sigma_E$ in $O(1)$ time.
\end{proof}

As a final note, we mention that for preprocessing the data structure from Lemma~\ref{lemma:hashing_multi_level_color_search}, we can perform a simple DFS traversal of each block $B_k$, and determine, for each color $c$ in block $B_k$, the respective strings $\sigma_V$ and $\sigma_E$, thus requiring $O(n)$ time. Then, the bottleneck of our data structure becomes the perfect hashing data structure, for which we know that an expected $O(n)$ time preprocessing is possible. 

We emphasize a crucial structural distinction here: when defining the edge marking $\sigma_E$ for the small tree $\tau$, we restrict our attention \textit{strictly} to the simple paths $P(u,v)$ between adjacent $1$-block representatives, treating these paths as single edges. We explicitly disregard the broader tree structure $\mathcal{T}(u,v)$ associated with these paths (as defined in Lemma~\ref{lemma:adjacent-representatives-connection}). This isolation of the backbone path from its surrounding subtree guarantees that $\sigma_E$ only flags colors that actually bridge the blocks, a distinction that will prove vital for the subsequent endpoint retrieval. In the spirit of Lemma~\ref{lemma:hashing_multi_level_color_search}, we now develop a data structure that allows us to precisely locate these occurrences.

\begin{lemma}
    \label{lemma:block-or-path-color-finding}
    There exists a linear-space data structure that can, in $O(1)$ time, answer the following queries:
    \begin{enumerate}
        \item Given a color $c$ and a $k$-block $B_k$, retrieve $u_{high}^c$, the highest occurrence of $c$ in $B_k$.
        \item Given a color $c$ and two adjacent $k$-block representatives $u$ and $v$, retrieve $u_c$ and $v_c$, defined as the occurrences of $c$ closest to $u$ and to $v$, respectively, on the path $P(u,v)$.
    \end{enumerate}
\end{lemma}
\begin{proof}
    \begin{figure}[H]
    \centering
    \begin{tabular}{c}
        \includegraphics[width=0.6\linewidth]{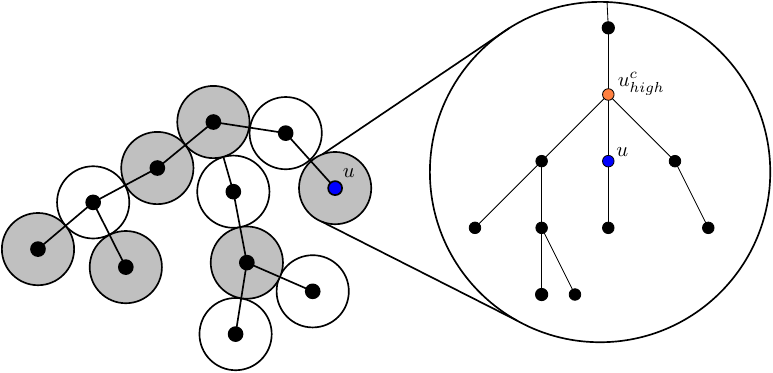}
    \end{tabular}
    \caption{An example blocking of a tree $\mathcal{T}$ with a vertex marking. A block is gray if it contains color $c$, or white otherwise. For each block $B_k$ containing $c$, our dictionary maintains the tuple $(c, B_k, u_{high}^c)$.}
    \label{fig:filled_smaller_blocks}
    \end{figure}
    
    For the block search, we build a dictionary $S_{\mathcal{B}_k}$ of valid tuples. Let $S_{\mathcal{B}_k}$ be the set of triples $(c, B_k, u_{high}^c)$. The total size of $S_{\mathcal{B}_k}$ is bounded by $O(n\log n)$ bits (which is $O(n)$ words), as each pair $(c, B_k)$ stores an $O(\log n)$-bit node identifier. We store this using static perfect hashing, allowing $O(1)$ lookups to retrieve $u_{high}^{c}$.

    \begin{figure}[H]
    \centering
    \begin{tabular}{c}
        \includegraphics[width=0.6\linewidth]{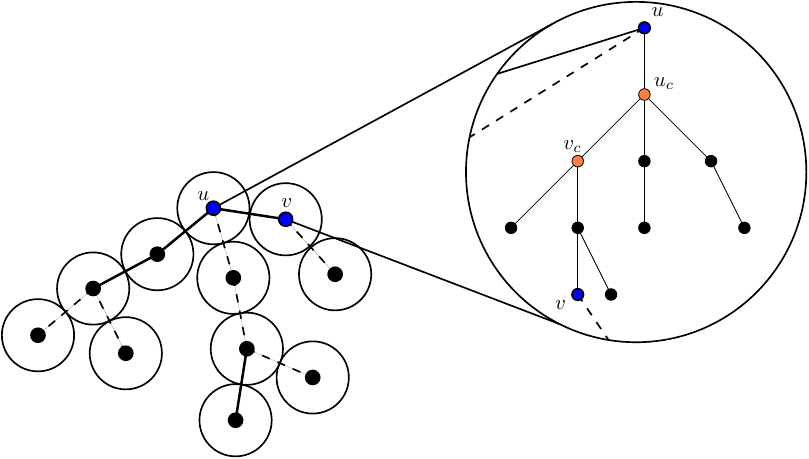}
    \end{tabular}
    \caption{An example blocking of a tree $\mathcal{T}$ with an edge marking. A path is a solid line if it contains color $c$, or dashed otherwise. For each path $P(u,v)$ containing $c$, our dictionary maintains the tuple $(c, u, v, u_c, v_c)$.}
    \label{fig:edge-containing-color-endpoint-occurrences}
    \end{figure}
    
    For retrieving occurrences on the extremities of a path $P(u,v)$ between adjacent $k$-block representatives, we build a dictionary $S_{E_k}$ containing tuples $(c, u, v, u_c, v_c)$. We can query $S_{E_k}$ using the triple $(c,u,v)$ to directly retrieve $u_c$ and $v_c$. Because the number of nodes occurring on any path between adjacent representatives sums to at most $O(n)$ across the tree, the number of distinct valid tuples is at most $O(n)$. Static perfect hashing again guarantees $O(1)$ retrieval.
\end{proof}

Combining these lemmas with the block tree definitions yields the following result for locating path endpoints.

\begin{corollary}
    \label{corollary:first_ocurrences_on_path_constant}
    Let $i_c$ and $j_c$ denote the first occurrences of a color $c$ on the query path $P(i, j)$ closest to $i$ and $j$, respectively. There exists a linear-space data structure that, given endpoints $i, j$, blocks $B_k(i), B_{k'}(j)$ ($k, k' \in \{2,3\}$), and a color $c$ present in the blocks but not in their predecessors $B_{k-1}(i), B_{k'-1}(j)$, determines $i_c$ (and symmetrically $j_c$) in $O(1)$ time.
\end{corollary}
\begin{proof}
    We assume the data structures from Lemmas~\ref{lemma:small_tree_search}, \ref{lemma:hashing_multi_level_color_search}, and \ref{lemma:block-or-path-color-finding}, alongside Corollary~\ref{corollary:smalle_tree_edge_search}, are constructed. Without loss of generality, we focus on routing the data structure to find $i_c$.

    First, we query the hashing structure from Lemma~\ref{lemma:hashing_multi_level_color_search} with $(c, B_k(i))$ to retrieve the node-marking $\sigma_V$ and edge-marking $\sigma_E$ for the block $B_k(i)$. Treating the hierarchy of $t_{k-1}$-blocks inside $B_k(i)$ as a small tree $\tau$, we use the structure from Lemma~\ref{lemma:small_tree_search} with start node $u = B_{k-1}(i)$ and bitstring $\sigma_V$ to verify that $c$ does not occur inside the predecessor block $B_{k-1}(i)$ but is present in the larger block $B_k(i)$. 

    Using Lemma~\ref{lemma:block-or-path-color-finding}, we retrieve $u^c_{high}$, the highest occurrence of color $c$ inside $B_k(i)$. Next, we identify the edge $e$ in $\tau$ such that $low(e)=u$, and query the edge search structure (Corollary~\ref{corollary:smalle_tree_edge_search}) with $(\tau, \sigma_E, e)$ to determine $e_c$, the lowest edge above $e$ containing $c$. If $e_c$ exists, we determine $v_{e}^c$, the lowest occurrence of $c$ on $e_c$, again using Lemma~\ref{lemma:block-or-path-color-finding}.

    \begin{figure}[H]
    \centering
    \hspace*{-2.3cm}
    \begin{tabular}{ccc}
        \includegraphics[width=0.4\linewidth]{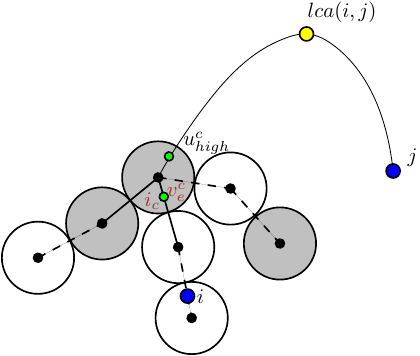} &
        \includegraphics[width=0.4\linewidth]{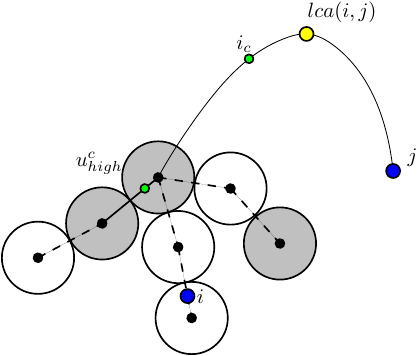} &
        \includegraphics[width=0.4\linewidth]{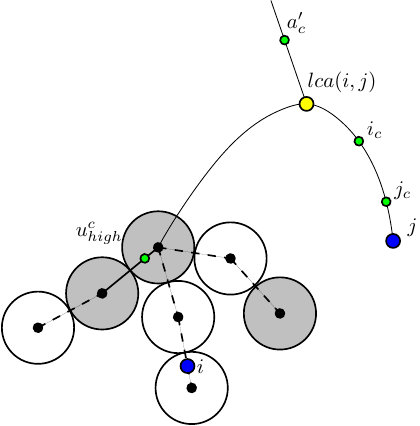} \\
        (a) $i_c$ found directly & (b) $i_c$ found as ancestor & (c) $i_c \in P(lca(i,j), j)$
    \end{tabular}
    \caption{The three topological cases for finding $i_c$, the first occurrence of color $c$ on $P(i,j)$. Solid black nodes represent occurrences of $c$. 
    In \textbf{(a)}, the data structures locate $v_e^c$ or $u_{high}^c$ directly on the upward path $P(i, lca(i,j))$, instantly identifying $i_c$. In this case, $i_c=v_e^c$. 
    In \textbf{(b)} Here $e_c$ is inexistent, and $u_{high}^c$ lies on side-branches, but the first color-$c$ ancestor of $u_{high}^c$ intersects the upward path $P(i, lca(i,j))$, identifying it as $i_c$. 
    In \textbf{(c)}, $c$ is absent from the upward path. We identify $a'_c$, the lowest out-of-path ancestor. By finding $j_c$ on the downward path and traversing $\delta_c - 1$ steps up the virtual tree $\mathcal{V}_c$ (dashed line), we land exactly on $i_c$.}
    \label{fig:path_endpoints_cases}
    \end{figure}

    By evaluating the topological positions of $u_{high}^c$ and $v_e^c$ relative to the path $P(i, lca(i,j))$, we can determine $i_c$ in $O(1)$ time. Specifically, if either node lies on $P(i, lca(i,j))$, the one closest to $i$ is exactly $i_c$.

    If neither node lies on $P(i, lca(i,j))$, we conclude that $c$ does not appear on the upward segment of the query path. In this case, we identify $a'_c$, the lowest color-$c$ ancestor of $u_{high}^c$ that lies outside $P(i,j)$. The true occurrence $i_c$ must therefore reside on the downward segment $P(lca(i,j), j)$.
    
    To resolve this scenario, we repeat the search process from endpoint $j$ to find $j_c$. If $j_c$ does not exist, $c$ is absent from $P(i,j)$ entirely. If $j_c$ does exist, we calculate the depth difference in the virtual tree $\mathcal{V}_c$: $\delta_c = depth_{\mathcal{V}_c}(j_c) - depth_{\mathcal{V}_c}(a_c')$. Because $a'_c$ is the lowest color-$c$ ancestor outside $P(i,j)$, the $(\delta_c-1)$-th ancestor of $j_c$ in the virtual tree $\mathcal{V}_c$ is guaranteed to lie on $P(i,j)$. Since all occurrences on this segment are ancestors of $j$, this $(\delta_c-1)$-th ancestor is exactly $i_c$. We retrieve it in $O(1)$ time using the Level Ancestor structure on $\mathcal{V}_c$.
\end{proof}

\subsection{Solving Subtask 5: Level 3 Intersection}
We apply a stratified blocking approach. We define a sequence of blocking factors $s_0, \dots, s_{|s|}$ such that $s_0 = t_3$, $s_{|s|} = t_2$, and $s_{k+1} = t_2 \cdot (\log(s_k/t_2))^2$. Note that $|s| = O(\log^* (t_3/t_2))$.
 
For each blocking factor $s_k$, let $\mathcal{B}_{s_k}$ be the set of $s_k$ blocks which partitions $V(\mathcal{T})$. For an $s_k$ block $B\in \mathcal{B}_{s_k}$, let $L_B$ denote the list of all nodes from $V(\mathcal{T})$ contained inside block $B$, ordered in increasing order of the node index. We construct tables $T_k$ for each level $s_k$. For every configuration of blocks $(B_{s_k}(i), B_{s_k}(j))$, we store the approximate position (using the most significant $2 \log(s_k/t_2)$ bits) of the optimal color $c$ within the linearized list of nodes of the blocks $B_{s_k}(i), B_{s_k}(j)$. The total space required is:
\[
    \sum_{k=0}^{|s|} \left(\frac{n}{s_k}\right)^2 \cdot 2 \log \left(\frac{s_k}{t_2}\right) = \left(\frac{n}{t_2}\right)^2 \sum_{k=0}^{|s|} \frac{2 \log r_k}{r_k^2} = O\left(\left(\frac{n}{t_2}\right)^2\right) \text{ bits,}
\]
where $r_k = s_k/t_2$. The query iterates through $O(t_2)$ candidate colors suggested by these tables. For each candidate $c$, we verify it is in $C(B_3) \cap \overline{C(B_2)}$ using Lemma~\ref{lemma:hashing_multi_level_color_search} in $O(1)$ time, find its endpoints $l, r$ using Corollary~\ref{corollary:first_ocurrences_on_path_constant}, and add $(c, l, r)$ to $S_2$.

\subsection{Solving Subtasks 6, 9, and 10}
\textbf{Subtask 6 (Mixed Level 2/3):} We iterate through all $c \in C(B_2(i))$. We verify membership in the target set $C_6$ using Lemma~\ref{lemma:hashing_multi_level_color_search} and find endpoints $(l, r)$ using Corollary~\ref{corollary:first_ocurrences_on_path_constant}, all in $O(1)$ per color. Valid triplets are added to $S_2$. The complexity is dominated by the size of $B_2(i)$, i.e., $O(t_2)$.

\textbf{Subtask 9 (Level 2 Intersection):} Analogous to Subtask 6; runtime is $O(t_2)$.

\textbf{Subtask 10 (Local):} We iterate through every color in $C(B_1(i)) \cup C(B_1(j))$ and add them directly to $S_1$ (verified in $O(\log \log n)$). The cost is $O(t_1)$.

Finally, to determine the actual answer to the query, we need to iterate through the colors inside $S_1$, find their endpoints on the path $P(i,j)$ by using the $O(\log\log n)$ time query from Lemma~\ref{lemma:durocher:path_frequency_query}, and compute the associated value of $g$. Then, we iterate through the triplets $(c,l,r)\in S_2$, and compute $g(c,l,r)$. The answer to the query for the path $P(i,j)$ is the color corresponding to the maximum value of $g$ computed in this step.

Processing $S_2$ will take $O(|S_2|)=O(t_2)$ time, and processing $S_1$ will take $O(|S_1|\cdot\log\log n)=O(t_1\cdot \log\log n)$ time. Since $t_1\log\log n=t_2$, the total time of processing the sets of candidates $S_2$ and $S_1$ is $O(t_2)$. We will present the preprocessing algorithms and analyze their respective runtime for each data structure in Appendix~\ref{appendix:precomputation_algorithms}. Now we only mention that a preprocessing runtime of $O(n(n/t_2))=O(n\cdot \frac{n}{t_1\log\log n})$ can be achieved. The total space required by our data structure is $O((n/t_2)^2)$ bits. By substituting $t_2$ by $t$, we formalize our result as follows.

\begin{lemma}
    \label{lemma:path_g_max_value_parameter_t}
    Given a colored tree $\mathcal{T}$, of $n$ nodes, a function $g$ respecting the properties from Section~\ref{section:problem_presentation}, and a blocking factor $t$, there exists a data structure, supporting path maximum $g$-value color queries in $O(t)$ time, requiring $O((\frac{n}{t})^2)$ space in bits, and an $O(n(n/t))$ preprocessing time.
\end{lemma}

By setting $t_1=\frac{\sqrt{n/w}}{\log\log n}$, while keeping the values of $t_2$ and $t_3$ the same as a function of $t_1$, we conclude this section with our main result directly following Lemma~\ref{lemma:path_g_max_value_parameter_t}.

\begin{theorem}
    \label{theorem:path_g_max_value_linear_space}
    Given a colored tree $\mathcal{T}$, of $n$ nodes, a function $g$ respecting the properties from Section~\ref{section:problem_presentation}, there exists a linear-space data structure, supporting path maximum $g$-value color queries in $O(\sqrt{n/w})$ time, requiring an $O(n\sqrt{nw})$ preprocessing time.
\end{theorem}

\section{\boldmath Faster Path \texorpdfstring{$\alpha$}{alpha}-Minority Queries on Trees}
\label{section:alpha_minority_in_alpha-1_time}

An $\alpha$-minority in a multiset $A$, for some parameter $\alpha \in [0,1]$, is defined as an element that occurs at least once in $A$ but constitutes no more than an $\alpha$ fraction of the total size of $A$. The path $\alpha$-minority query for a tree path $(u,v)$ asks to report any single color that occurs at most an $\alpha$ fraction of the tree path between nodes $u$ and $v$. Durocher et al.~\cite{durocher2016linear} introduced a linear-space data structure capable of answering path $\alpha$-minority queries on trees in $O(\alpha^{-1}\log\log n)$ time, where the parameter $\alpha$ is specified at query time, and can differ for every query. In this section, we show that by incorporating randomness, the query time can be improved to $O(\alpha^{-1})$ for a Monte Carlo algorithm (success probability $\geq 1/2$) and to expected $O(\alpha^{-1} + \log\log n)$ for a Las Vegas algorithm (success probability $1$).

\subsection{\boldmath \texorpdfstring{$O(\alpha^{-1})$}{O(alpha-1)} Time Algorithm with Success Probability \texorpdfstring{$\geq 1/2$}{geq1/2}}

We use the following result from Durocher et al.~\cite{durocher2016linear} to efficiently retrieve candidate colors near the endpoints of the path.

\begin{lemma}[Durocher et al.~\cite{durocher2016linear}]
    There exists an $O(n)$-space data structure that supports $k$-nearest distinct ancestor queries on trees in $O(k)$ time. The query returns the $k$ distinct colors closest to a node $u$ on the path to the root, in order of increasing distance.
\end{lemma}

Let $P(i, j)$ be the query path, and let $w = \text{LCA}(i, j)$. We define the length of the path as $N = |P(i, j)|$. We propose the following $4$-phase query algorithm:

\begin{itemize}
    \item \textbf{Phase 1 (Candidate Collection):}
    We retrieve the sets of distinct colors closest to the endpoints. Let $k = \lceil 2\alpha^{-1} \rceil$.
    \begin{align*}
        S_l &\gets \text{the } k \text{ distinct colors on } P(i, w) \text{ closest to } i. \\
        S_r &\gets \text{the } k \text{ distinct colors on } P(j, w) \text{ closest to } j.
    \end{align*}
    Note that if the path segment $P(i, w)$ contains fewer than $k$ distinct colors, $S_l$ will simply contain all distinct colors on that segment (and similarly for $S_r$).

    \item \textbf{Phase 2 (Pruning Majorities):}
    We remove obvious majorities. For each $c \in S_l$, if the frequency of $c$ on the sub-path $P(i, w)$ exceeds $\alpha N$, we remove $c$ from $S_l$. We perform the symmetric check for $S_r$ on $P(j, w)$. This is valid because if a color is a majority on a sub-path with frequency exceeding the threshold for the \textit{entire} path, it is certainly an $\alpha$-majority for the whole path.

    \item \textbf{Phase 3 (Exact Verification of Overlap):}
    Identify colors $c \in S_l \cap S_r$. For these colors, we have identified both their first occurrence from $i$ and their first occurrence from $j$. Since the sets $S_l$ and $S_r$ are built searching inwards from the endpoints, if $c$ appears in both, we have bounded its occurrences. We can compute the exact frequency of $c$ on $P(i, j)$ in $O(1)$ time (using the precomputed frequency data structures). If the frequency is $\leq \alpha N$, we return $c$ immediately as the answer. We remove all colors $c\in S_l\cap S_r$ from $S_l$ and $S_r$ respectively. Thus, at the end of this phase, $S_l\cap S_r=\emptyset$. 

    \item \textbf{Phase 4 (Random Sampling):}
    If no answer was returned in Phase 3, we construct the set of remaining candidates $D = S_l \cup S_r$. We select a color $c$ uniformly at random from $D$ and return it.
\end{itemize}

\begin{lemma}
    The query algorithm returns an $\alpha$-minority in $O(\alpha^{-1})$ time with probability $\ge \frac{1}{2}$.
\end{lemma}
\begin{proof}
    The time complexity is dominated by the retrieval of $O(\alpha^{-1})$ colors and the set operations, which take $O(\alpha^{-1})$ time. The frequency checks in Phase 2 and 3 take $O(1)$ using the standard Level Ancestor and precomputed depth/frequency arrays (as employed in the virtual tree structures).

    We analyze the success probability in two cases:
    \begin{enumerate}
        \item \textbf{Case 1: The path has few distinct colors.}
        Suppose the number of distinct colors on $P(i, j)$ is at most $2\lceil 2\alpha^{-1} \rceil$. In this scenario, $S_l$ and $S_r$ effectively cover the entire path. Any color present on the path is in $S_l \cup S_r$. Specifically, any $\alpha$-minority is in this set. Thus, the algorithm succeeds with probability $1$.

        \item \textbf{Case 2: The path has many distinct colors.}
        Suppose there are more than $2\lceil 2\alpha^{-1} \rceil$ distinct colors. The number of $\alpha$-majorities (colors with frequency $> \alpha N$) is strictly less than $\alpha^{-1}$.
        Our candidate set $D = S_l \cup S_r$ has size up to $2 \lceil 2\alpha^{-1} \rceil \approx 4\alpha^{-1}$.
        Even in the worst case where every $\alpha$-majority is present in $D$, there are at most $\alpha^{-1}$ "bad" candidates.
        The number of "good" candidates ($\alpha$-minorities) in $D$ is at least $|D| - \alpha^{-1} \approx 3\alpha^{-1}$.
        Therefore, the probability of selecting an $\alpha$-minority at random is at least: $\frac{3\alpha^{-1}}{4\alpha^{-1}} = \frac{3}{4} \ge \frac{1}{2}.$
    \end{enumerate}
\end{proof}

\subsection{\boldmath Expected \texorpdfstring{$O(\alpha^{-1} + \log\log n)$}{O(alpha-1+loglog n)} Time Algorithm}

To achieve a success probability of $1$, we extend Phase 4 into a Las Vegas algorithm. Instead of returning the random candidate immediately, we verify it. We can determine the exact endpoints of $c$ on $P(i, j)$ using an $O(\log \log n)$-time predecessor/successor query (using the structure from Lemma~\ref{lemma:durocher:path_frequency_query}). Once the endpoints are known, calculating the frequency takes $O(1)$ time. We repeatedly pick a color $c$ at random from $S_l\cup S_r$, and verify its frequency, until an $\alpha$-minority is found. 

Since the probability of success in one sampling step is at least $1/2$, the number of trials follows a geometric distribution with success parameter $p \ge 1/2$. The expected number of trials is at most $2$.
The total expected time complexity is: $O(\alpha^{-1}) + E[\text{trials}] \cdot O(\log \log n) = O(\alpha^{-1} + \log \log n).$



\bibliography{main}

@article{chan2014linear,
  title={Linear-space data structures for range mode query in arrays},
  author={Chan, Timothy M and Durocher, Stephane and Larsen, Kasper Green and Morrison, Jason and Wilkinson, Bryan T},
  journal={Theory of Computing Systems},
  volume={55},
  number={4},
  pages={719--741},
  year={2014},
  publisher={Springer}
}

@inproceedings{chan2012linear,
  title={Linear-space data structures for range minority query in arrays},
  author={Chan, Timothy M and Durocher, Stephane and Skala, Matthew and Wilkinson, Bryan T},
  booktitle={Scandinavian Workshop on Algorithm Theory},
  pages={295--306},
  year={2012},
  organization={Springer}
}

@article{durocher2016linear,
  title={Linear-space data structures for range frequency queries on arrays and trees},
  author={Durocher, Stephane and Shah, Rahul and Skala, Matthew and Thankachan, Sharma V},
  journal={Algorithmica},
  volume={74},
  pages={344--366},
  year={2016},
  publisher={Springer}
}

@inproceedings{brodnik1999resizable,
  title={Resizable arrays in optimal time and space},
  author={Brodnik, Andrej and Carlsson, Svante and Demaine, Erik D and Ian Ian Munro, J and Sedgewick, Robert},
  booktitle={Algorithms and Data Structures: 6th International Workshop, WADS’99 Vancouver, Canada, August 11--14, 1999 Proceedings 6},
  pages={37--48},
  year={1999},
  organization={Springer}
}

@article{krizanc2005range,
  title={Range mode and range median queries on lists and trees},
  author={Krizanc, Danny and Morin, Pat and Smid, Michiel},
  journal={Nordic Journal of Computing},
  volume={12},
  number={1},
  pages={1--17},
  year={2005},
  publisher={Publishing Association Nordic Journal of Computing}
}

@inproceedings{hagerup2001efficient,
  title={Efficient minimal perfect hashing in nearly minimal space},
  author={Hagerup, Torben and Tholey, Torsten},
  booktitle={STACS 2001: 18th Annual Symposium on Theoretical Aspects of Computer Science Dresden, Germany, February 15--17, 2001 Proceedings 18},
  pages={317--326},
  year={2001},
  organization={Springer}
}

@inproceedings{belazzougui2009hash,
  title={Hash, displace, and compress},
  author={Belazzougui, Djamal and Botelho, Fabiano C and Dietzfelbinger, Martin},
  booktitle={European Symposium on Algorithms},
  pages={682--693},
  year={2009},
  organization={Springer}
}

@article{10.1145/828.1884,
author = {Fredman, Michael L. and Koml\'{o}s, J\'{a}nos and Szemer\'{e}di, Endre},
title = {Storing a Sparse Table with 0(1) Worst Case Access Time},
year = {1984},
issue_date = {July 1984},
publisher = {Association for Computing Machinery},
address = {New York, NY, USA},
volume = {31},
number = {3},
issn = {0004-5411},
url = {https://doi.org/10.1145/828.1884},
doi = {10.1145/828.1884},
journal = {J. ACM},
month = jun,
pages = {538–544},
numpages = {7}
}

@article{bender2004level,
  title={The level ancestor problem simplified},
  author={Bender, Michael A and Farach-Colton, Mart{\i}n},
  journal={Theoretical Computer Science},
  volume={321},
  number={1},
  pages={5--12},
  year={2004},
  publisher={Elsevier}
}

@article{berkman1994finding,
  title={Finding level-ancestors in trees},
  author={Berkman, Omer and Vishkin, Uzi},
  journal={Journal of computer and System Sciences},
  volume={48},
  number={2},
  pages={214--230},
  year={1994},
  publisher={Elsevier}
}

@inproceedings{dietz1991finding,
  title={Finding level-ancestors in dynamic trees},
  author={Dietz, Paul F},
  booktitle={Workshop on Algorithms and Data Structures},
  pages={32--40},
  year={1991},
  organization={Springer}
}

@article{berkman1993recursive,
  title={Recursive star-tree parallel data structure},
  author={Berkman, Omer and Vishkin, Uzi},
  journal={SIAM Journal on Computing},
  volume={22},
  number={2},
  pages={221--242},
  year={1993},
  publisher={SIAM}
}

@inproceedings{bender2000lca,
  title={The LCA problem revisited},
  author={Bender, Michael A and Farach-Colton, Martin},
  booktitle={Latin American Symposium on Theoretical Informatics},
  pages={88--94},
  year={2000},
  organization={Springer}
}

@article{chazelle1986filtering,
  title={Filtering search: A new approach to query-answering},
  author={Chazelle, Bernard},
  journal={SIAM Journal on Computing},
  volume={15},
  number={3},
  pages={703--724},
  year={1986},
  publisher={SIAM}
}

@Inbook{Skala2013,
author="Skala, Matthew",
title="Array Range Queries",
bookTitle="Space-Efficient Data Structures, Streams, and Algorithms: Papers in Honor of J. Ian Munro on the Occasion of His 66th Birthday",
year="2013",
publisher="Springer Berlin Heidelberg",
address="Berlin, Heidelberg",
pages="333--350",
isbn="978-3-642-40273-9",
doi="10.1007/978-3-642-40273-9_21",
url="https://doi.org/10.1007/978-3-642-40273-9_21"
}

\appendix
\section{Building the Virtual Trees}
\label{appendix:virtual_trees}

In this section, we provide the algorithms and complexity analysis for computing all the virtual trees $\mathcal{V}_c$. We will use the following result of Brodnik et al.~\cite{brodnik1999resizable} on implementing stacks using resizable arrays:

\begin{lemma}[Brodnik et al.~\cite{brodnik1999resizable}]
    Using singly resizable arrays, stacks can be implemented in $O(1)$ worst-case time per operation, and $O(\sqrt{n})$ extra storage, besides the storage necessary to maintain the contents of the stack.
\end{lemma}

 Prior to the DFS procedure, for each color $c$, we instantiate an empty stack $\mathcal{S}_c$. For each color $c$ for which the virtual tree $\mathcal{V}_c$ is not connected by default, we push the virtual root node $r_c$ to the stack $\mathcal{S}_c$, as described in section~\ref{section:prerequisites}.

 We will perform a single DFS traversal over the input tree $\mathcal{T}$ starting from node $1$, which is also the root of $\mathcal{T}$. For each color $c$, we build the parent list $parent_{c}[index_c(u)]$ for each node $u$ of color $c$. By $index_c(u)$ we denote an index assigned to the node $u$ of color $c$ in some arbitrary ordering of the nodes of color $c$. The arrays $parent_c$ and $index_c$ can be initialized in $O(n)$ time, and $index_c$ can also be precomputed in $O(n)$ time. Each time the DFS procedure enters a node $u$ of color $c$, we query the top of $\mathcal{S}_c$, and assign $parent_c[index_c(u)]\gets \mathcal{S}_c.top()$. This is exactly what our definition of $\mathcal{V}_c$ implies, as the parent of node $u$ in $\mathcal{V}_c$, is the last node visited in a DFS traversal before entering $u$. When the DFS traversal leaves the subtree of node $u$, we pop the node $u$ from $\mathcal{S}_c$. The pseudocode for the DFS traversal can be summarized the following way:

\begin{algorithm}[H]
\caption{DFS on $\mathcal{T}$ for building the virtual trees}
\begin{algorithmic}[1]
\State Input: the currently examined node $u$
\State $c\gets \text{color of $u$}$
\State $parent_c[index_c(u)]\gets \mathcal{S}_c.top()$
\State $\mathcal{S}_c.push(u)$
\For{$v\in children(u)$}
    \State DFS($v$)
\EndFor
\State $\mathcal{S}_c.pop(section:virtual_trees)$
\end{algorithmic}
\end{algorithm}

It is well known that the children arrays of each node of a tree can be computed in linear time from its parent array, thus we arrive at a $O(n)$ time algorithm for computing all the virtual trees $\mathcal{V}_c$, which, in their turn, require $O(n)$ words of space in total.
\section{Proofs for Blocking Results}
\label{appendix:blocking_results_proofs}

In this section, we provide the implementation details and complexity analysis for the blocking strategies introduced in Section~\ref{section:blocking_technique}. While the existence of these structures was established by Durocher et al.~\cite{durocher2016linear}, we formalize the algorithms here to prove they can be constructed in $O(n \log n)$ time.

First, we define the \textbf{Auxiliary Tree}, a structure implicitly used in the second step of the marking procedure to ensure LCA closure.

\subsection{Auxiliary Trees}
\label{subsection:auxiliary_trees}

The \textbf{Auxiliary Tree} (often referred to in literature as a Virtual Tree, distinct from the color-based virtual trees in Subsection~\ref{subsection:prerequisites:subsubsection:intermediary_results_on_colored_trees}) is a compression of the original tree $\mathcal{T}$ with respect to a specific subset of marked nodes $S \subseteq V(\mathcal{T})$.

\begin{definition}
    Given a subset of marked nodes $S$, the Auxiliary Tree $\mathcal{T}_{aux}(S)$ is the minimal subgraph of $\mathcal{T}$ that contains $S$ and preserves connectivity between them. Its vertex set consists of $S \cup \{ \text{LCA}(u, v) \mid u, v \in S \}$.
\end{definition}

\textbf{Construction.} The Auxiliary Tree can be constructed efficiently using the following standard procedure:
\begin{enumerate}
    \item Sort the nodes in $S$ according to their pre-order traversal index in $\mathcal{T}$.
    \item Compute the LCA of every adjacent pair of nodes in the sorted list.
    \item The vertex set of $\mathcal{T}_{aux}$ is the union of $S$ and these LCAs.
    \item Edges are formed by connecting each node to its nearest ancestor present in this new set using a stack-based construction.
\end{enumerate}

\textbf{Complexity.} If the pre-order indices and an $O(1)$ LCA oracle are precomputed (taking $O(n)$ time), constructing $\mathcal{T}_{aux}(S)$ takes $O(|S| \log |S|)$ time (dominated by sorting).

\subsection{Proofs of Computational Complexity}

\begin{lemma}[Durocher et al.~\cite{durocher2016linear}]
\label{lemma:durocher_blocking_runtime}
    Given a tree $\mathcal{T}$ with $n$ nodes and a blocking factor $t \le n$, the set of marked nodes $M_t$ obeying the properties of Lemma~\ref{lemma:durocher_blocking} can be computed in $O(n \log n)$ time.
\end{lemma}
\begin{proof}
    We employ the marking procedure exactly as defined by Durocher et al.~\cite{durocher2016linear}. The procedure consists of three specific rules applied to determine the set $M_t$:
    \begin{enumerate}
        \item Mark any node $u$ where $\text{size}(u) \ge t$ and $\text{size}(v) < t$ for every child $v$ of $u$.
        \item Mark the lowest common ancestor of any two marked nodes.
        \item As long as there is an unmarked node $v$ with a descendant $u$ such that $\text{size}(v) - \text{size}(u) \ge t$, and $v$ is the lowest such node in the tree, mark $v$.
    \end{enumerate}

    We implement these rules in three computational phases to establish the $O(n \log n)$ bound:

    \textbf{Phase 1: Size Threshold Marking (Rule 1).}
    We perform a single Depth First Search (DFS) traversal to compute subtree sizes. During the post-order visit, we identify and mark all nodes satisfying Rule 1, adding them to an initial set $S_1$. This takes $O(n)$ time.

    \textbf{Phase 2: LCA Closure (Rule 2).}
    To satisfy Rule 2, we compute the LCA-closure of $S_1$. Let $S_2$ be the set of nodes in the \textbf{Auxiliary Tree} constructed from $S_1$. This structure implicitly includes all necessary LCAs. Sorting $S_1$ by pre-order index takes $O(|S_1| \log |S_1|)$. Since $|S_1| \le n/t$, this step takes $O(n \log n)$.

    \textbf{Phase 3: Gap Filling (Rule 3).}
    Rule 3 requires filling gaps where the unmarked path length exceeds $t$. The condition $\text{size}(v) - \text{size}(u) \ge t$ implies a path of unmarked nodes with a total weight (subtree size difference) exceeding the threshold.
    \begin{enumerate}
        \item We interpret the edges of the Auxiliary Tree (from Phase 2) as paths of unmarked nodes.
        \item We maintain a max-priority queue of these paths, keyed by the size difference between endpoints.
        \item While the largest gap violates the condition, we extract it, identify the lowest node $v$ satisfying Rule 3, mark it, and update the Auxiliary Tree and priority queue.
    \end{enumerate}
    The total number of marked nodes is $O(n/t)$. Each marking operation involves priority queue updates taking $O(\log n)$. Thus, this phase runs in $O((n/t) \log n) \subseteq O(n \log n)$.

    Thus, the total runtime of this precomputation algorithm is $O(n \log n)$.
\end{proof}

\begin{lemma}[Durocher et al.~\cite{durocher2016linear}]
\label{lemma:durocher_real_blocking_runtime}
    Given a tree $\mathcal{T}$ and the set of marked nodes $M_t$ from Lemma~\ref{lemma:durocher_blocking_runtime}, the partition of $V(\mathcal{T})$ into blocks satisfying the properties of Lemma~\ref{lemma:durocher_real_blocking} can be computed in $O(n \log n)$ time.
\end{lemma}
\begin{proof}
    We construct the partition using the specific assignment logic defined by Durocher et al.~\cite{durocher2016linear}. Their procedure defines two types of blocks:
    \begin{enumerate}
        \item \textbf{Internal Blocks ($B_u$):} For a marked node $u$, $B_u$ contains all nodes reachable from $u$ by a path that exits $u$ to its parent and contains no other marked nodes.
        \item \textbf{Leaf Blocks ($B'_w$):} For a marked node $w$ with no marked descendants, $B'_w$ contains all nodes in the subtree of $w$ not reachable from any other marked node (i.e., nodes "near the bottom" of the tree).
    \end{enumerate}

    We compute this partition efficiently as follows:

    \textbf{Step 1: Compute Marked Nodes.}
    Compute $M_t$ using the procedure from Lemma~\ref{lemma:durocher_blocking_runtime}.
    This takes $O(n \log n)$ time.

    \textbf{Step 2: Assign Leaf Blocks.}
    We identify marked nodes $w$ with no marked descendants. We perform a DFS from each such $w$ downwards. All visited unmarked nodes are assigned to $B'_w$. The traversal stops at leaves. Each node in these regions is visited once, taking $O(n)$ time.

    \textbf{Step 3: Assign Internal Blocks.}
    We identify the remaining unassigned nodes. For every marked node $u$, we traverse upwards (or essentially perform a restricted DFS on the forest formed by removing marked nodes). All nodes in the connected component adjacent to $u$'s parent edge are assigned to $B_u$. This is a standard graph traversal visiting every node and edge exactly once, taking $O(n)$ time.

    Dominated by Step 1, the total time is $O(n \log n)$.
\end{proof}
\section{Preprocessing Algorithms}
\label{appendix:precomputation_algorithms}

In this section, we provide the algorithms for computing the components of our data structure. We begin by designing a universal algorithm that is used as a subroutine for computing all the components. Once we establish the universal algorithm, designing the separate procedures for each of the components becomes trivial. 

\begin{lemma}
    \label{lemma:universal_precomputation_traversal}
    Given a colored tree $\mathcal{T}$, of $n$ vertices, and two blocking partitions $\mathcal{B}_{t_1}$ and $\mathcal{B}_{t_2}$, associated with blocking factors $t_1$ and $t_2$, such that $t_2/t_1<\frac{\log n}{8\log\log n}$, respectively, there exists an $O(n(n/t_1))$ time algorithm, requiring $O(n)$ work-space, which lists the optimal color inside the following sets of colors:
    \begin{itemize}
        \item $ c_{opt}\in C(P(u,v)) \cap ( C( B_{1}(u))  \cup C(B_1(v))) \cap \overline{C(B_2(u))} \cap \overline{C(B_2(v))}$
        The optimal color in the set corresponding to the colors appearing on the path $P(u,v)$, in block $B_1(u)$ \textbf{or} $B_1(v)$, and not present in block $B_2(u)$ or $B_2(v)$.
        
        \item $ c_{opt}\in C(P(u,v)) \cap C( B_{1}(u)  \cap C(B_1(v)) \cap \overline{C(B_2(u))} \cap \overline{C(B_2(v))}$

        The optimal color in the set corresponding to the colors appearing on the path $P(u,v)$, in block $B_1(u)$ \textbf{and} in block  $B_1(v)$, and not present in block $B_2(u)$ or $B_2(v)$.
    \end{itemize}
    where $u$, $v$ are representative nodes of $B_1(u)$ and $B_1(v)$. 
    This final list will have length $(n/t_1)^2$.
\end{lemma}
\begin{proof}
   We iterate through each pair of blocks $B_u \in \mathcal{B}_{t_1}$ and $B_v \in \mathcal{B}_{t_2}$, where $u$ is the representative of $B_u$ and $v$ is the representative of $B_v$.
    For the pair $(u,v)$, we determine the following important nodes:
    
    \begin{itemize}
        \item \textbf{The first $t_1$-block $B_{v'}$ on path $P(u,v)$, and node $v'$. } 
        
        By appealing to the block tree $\mathcal{T}_{t_2/t_1}[B_v]$ of $B_v$, we determine in $O(t_2/t_1)$ time, the $t_1$-block $B_{v'}\in \mathcal{B}_{t_1}$, which is the first $t_1$ block of $\mathcal{T}_{t_2/t_1}[B_v]$, encountered by traversing the path $P(u,v)$, from the representative node $u$ to node $v$. We assume $v'$ is the representative of $B_{v'}$. This will take $O(t_2/ t_1)$ time.
        \item \textbf{The first node $v''$ contained in $B_2(v)$, on path $P(u,v)$.}
        
        By iterating through the nodes of $\mathcal{T}$ inside $B_{v'}$, we determine the first node $v''$ contained in $B_2(v)$, encountered when traversing the path $P(u,v)$. This takes $O(t_1)$ time.
    \end{itemize}

    From $v''$, we start a DFS traversal through the subtree of $\mathcal{T}$ induced by the nodes present in $B_2(v)$. Suppose we arrived at a node $x$ of color $c$. We are only interested in the nodes $x$, corresponding to the following occurrences of $c$:
    
    \begin{enumerate}
        \item $x$ is the first occurrence of $c$, on the path $P( B_1(y), v')$, for some node $y$ inside the $t_2$-block $B_2(v)$.
        \item $x$ is the first occurrence of $c$, on the path $P(v', u)$.
    \end{enumerate}
    
    We refer to the node $x$ as being of type $1$ or $2$, according to the cases defined above. It is easy to develop a linear-space data structure that can determine whether $x$ is of type $1$ or $2$, or none of these types, in $O(1)$ time, hence we omit the proof.
    
    \textbf{Solving Case 1. }
        For a node $x$ of type $1$, we are interested in all nodes $y$, such that $x$ is the first occurrence of $c$, on the path $P( B_1(y), v')$. We limit our search to the nodes $y$, that are representatives of their corresponding $t_1$-block, that is, $y$ is the representative of $B_1(y)$. We denote the set of such nodes by $S_x$. Listing all nodes from $S_x$ would take $O(t_2/t_1)$ time, which is too slow. However, we can work around this limitation by employing tabulation and small trees once more. The following result is crucial for developing our tabulation approach:
        \begin{lemma}
        \label{lemma:tabulation_small_tree_subset_nodes_reaching_e}
            There exists a linear-space data structure, requiring an $O(n)$ preprocessing time, that supports the following queries for rooted small trees in $O(1)$ time: 
            
            \textbf{Input: }
            A small tree $\tau$, rooted at node $r$; a binary string $\sigma_V$ of length $|V(\tau)|$; a binary string $\sigma_E$ of length $|E(\tau)|$; and an edge $e$ of $E(\tau)$ such that $\sigma_E(e)=1$,
            
            \textbf{Output: } Return a binary string $\psi$, corresponding to the subset of nodes $S\subseteq V(\tau)$, marked with $0$ according to $\sigma_V$, such that, for any node $u\in S$, $e$ is the first node marked with $1$, on the path $P(u,r)$. 
        \end{lemma}
        \begin{proof}
        Similar to the proof of Lemma~\ref{lemma:small_tree_search}, we will construct a table $T$, whose entries will contain the answer to a specific query, $\psi= T[\tau][r][\sigma_V][\sigma_E][e]$.
        The table $T$ will have size $O(2^{s\log s}\cdot s\cdot 2^{s}\cdot 2^{s}\cdot s) \subset O(2^{2s\log s})\subset O(n^{1/4})\ll O(n)$.
        To compute the answer to a query, stored in an entry of this table, we require a $O(s)$ DFS traversal of the tree $\tau$, beginning with edge $e$, which keeps our runtime in $O(n^{1/4})\ll O(n)$.
        \end{proof}
        To use this tabulation technique, we require the additional computation of the binary strings $\sigma_E$. The binary strings $\sigma_V$ are the binary strings corresponding to the presence of a color $c$ among the blocks of $\mathcal{T}_{t_2/t_1}[B_2(v)]$. Meanwhile, the binary strings $\sigma_E$ correspond to the presence of a color $c$ on the path between the representatives of adjacent blocks. This can also be computed in $O(n)$ time, requiring an $O(n)$-space perfect hashing data structure.

            \begin{figure}[H]
    \centering
    \begin{tabular}{cc}
        \includegraphics[width=0.3\linewidth]{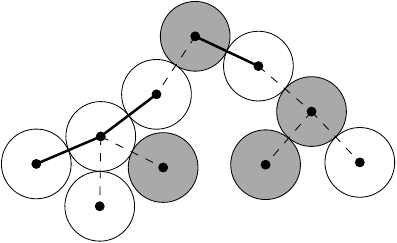} &
        \includegraphics[width=0.4\linewidth]{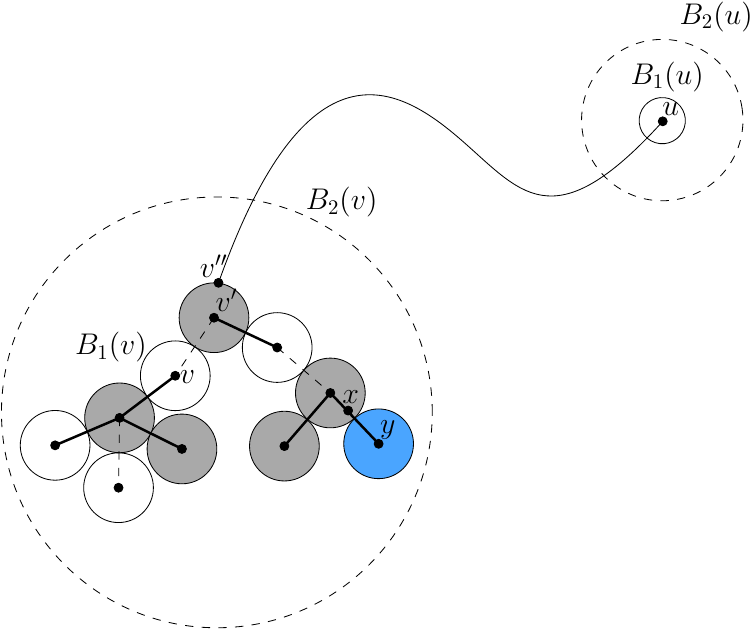} \\
        (a) & (b)
    \end{tabular}
    \caption{Small trees with vertex marking $\sigma_V$ and edge marking $\sigma_E$. A node is colored in gray if it is marked with $1$ and white if it is marked with $0$ according to $\sigma_V$. An edge is a normal segment if it is marked with $1$ and a dashed line if it is marked with $0$ according to $\sigma_E$. In figure (a), a small tree with some arbitrary markings $\sigma_V$ and $\sigma_E$ is depicted. In figure (b), a small tree $\tau$ together with the markings $\sigma_V$ $\sigma_E$, corresponding to the positions of some color $c$, is superimposed over a tree block $\mathcal{T}_{t_2/t_1}[B_2(v)]$. The nodes $u$, $v$, $v'$ and $v''$ are also depicted. A possible status node $x$ of the DFS procedure is depicted, and the only possible $y$ node of it is colored with blue.}
    \label{fig:small_trees_edge_colors}
    \end{figure}

        \textbf{Solving case 1.} To solve case $1$, we maintain for the $t_2$-block $B_2(v)$, a table $T$ (do not mistake it for the table $T$ from the proof of Lemma~\ref{lemma:tabulation_small_tree_subset_nodes_reaching_e}), maintaining for each binary string $\psi$, corresponding to a subset of nodes $S_x$, one optimal color, corresponding to a maximum value of the function $g$. That is, $T[\psi]=c_{opt}$.

        For the node $x$, we can determine the binary string $\psi_x$ corresponding to the set $S_x$ by an $O(1)$ query to the data structure from Lemma~\ref{lemma:tabulation_small_tree_subset_nodes_reaching_e}. At the same time, we must determine the occurrence of $c$, encountered first on the path $P(u,x)$, by querying the block search data structure from Lemma~\ref{lemma:hashing_multi_level_color_search}, for the big block $B_2(u)$ and small block $B_1(u)$. Assume this occurrence of $c$ appears at node $r$. Finally, we update the entry $T[\psi_x]$, if $g(c,x,r)$ is greater than the $g(c,v,u)$. Note that for achieving this, we need to store the maximum found value of $g$ inside a separate table.

        Building the table $T[\psi]$ takes $O(t_2)$ time, as we pass through each node inside $B_2(v)$ exactly once, and perform a constant number of $O(1)$-time operations on each node.

        We now use $T[\psi]$ to compute the answer for each small $t_1$-block inside $B_2(v)$. We iterate through every pair $(\psi, z)$, where $\psi$ is a binary string and $z$ is an index of a node of $\mathcal{T}_{t_2/t_1}[B_2(v)]$, such that $\psi[z]=1$. During this iterative process, we update the optimal values of $g$ and $c_{opt}$ for each index $z$, which we can store in a separate table. This procedure requires $O(2^{t_2/t_1}\cdot t_2/t_1)$ time.
        
    \textbf{Solving Case 2. } The solution to case $2$ is analogous to that of case $1$, the only difference being that now we need to construct a new small tree $\tau'$, by adding a new node $v_{new}$, and setting it as the parent of $v'$. We also set the binary symbol associated with the edge $e=(v_{new},v)$ to $1$ in the new associated $\sigma_{E}$. 

    One important scenario that we missed is when the first occurrence of an optimal color $c$ appears outside $B_2(v)$. This is not a problem, as we can use the solution to \textbf{Case 2} to process the colors falling under this assumption.

    The total runtime of processing a pair of blocks $(B_u,B_v)$ is $O(t_2+2^{t_2/t_1}\cdot t_2/t_1)$. The total number of such pairs is $(n/t_1)\cdot (n/t_2)$. Thus, the total runtime of this algorithm is $O(n(n/t_1))$, 
    assuming that $t_2/t_1$ is reasonably small.
\end{proof}

In our case, we deal with a partition hierarchy, and it is easy to see that the precomputation algorithm provided in the proof of Lemma~\ref{lemma:universal_precomputation_traversal} can be adapted to this case. We can modify the algorithm such that it computes the optimal color and its corresponding value of $g$, for pairs of blocks $(B_u, B_v)$, where $B_u$ is a $k$-level block and $B_v$ is a $k'$-level block, while keeping the runtime of the algorithm at $O(n(n/t_k))$.

Now we refer to the subtask decomposition from Subsection \ref{subsection:decomposition_into_subtasks}. First, we present the precomputation algorithm for the table used in Subtask 1.

\textbf{Precomputation for Subtask 1. } In order to compute the table required in subtask 1, we need to introduce one more blocking factor to our multi-level partition. Namely, we introduce the blocking factor $t_4=t_3\sqrt{\log n}$.

We aim to determine the optimal color from the following sets of colors:
\begin{itemize}
    \item \textbf{Case 1}: $C_{path} \cap C(B_4(u))\cap C(B_4(v)) \cap \overline{C(B_3(u))} \cap \overline{C(B_3(v))}$
    \item \textbf{Case 2}:  $C_{path} \cap C(B_3(v)) \cap \overline{C(B_4(u))} \cap \overline{C(B_4(v))}$
    \item \textbf{Case 3}: $C_{path} \cap \overline{C(B_4(u))} \cap \overline{C(B_4(v))}$
\end{itemize}

\textbf{Solving case 1.} According to Lemma~\ref{lemma:universal_precomputation_traversal}, we can build the table $T$, such that, for every pair $(B_u,B_v)$, where $B_u$ and $B_v$ are $t_3$ blocks, with their respective representative nodes $u$ and $v$, the entry $T[B_u][B_v]$ stores the optimal color $c_{opt}$ for the examined set of colors.

\textbf{Solving case 2.} To compute the optimal color for the sets corresponding to \textbf{Case 2}, we can again apply the algorithm from Lemma~\ref{lemma:universal_precomputation_traversal}, with a slight modification. The algorithm requires an $O(1)$-time oracle for computing the first occurrence of a color $c$ on the path $P(u,v)$, which appears inside $B_{4}(v)$ and outside $B_3(v)$. Instead of using an $O(1)$-time oracle, we will apply the $O(\log\log n)$-time search from Lemma~\ref{lemma:durocher:path_frequency_query}. This increases the preprocessing time of our algorithm from $O(n\frac{n}{t_4})$ to $O(n(n/t_4)\log\log n)\subset O(n(n/t_3))$.

\textbf{Solving case 3.} To compute the optimal color for the sets corresponding to \textbf{Case 3}, we can implement the trivial algorithm for building the respective table $T_2$, such that $T_2[B_u][B_v]$ contains the optimal color from the set $C_{path} \cap \overline{C(B_4(u))} \cap \overline{C(B_4(v))}$. For each $t_4$-block $B_u\in \mathcal{B}_{t_4}$, we start a DFS traversal of $\mathcal{T}$, starting from $u$, the representative of $B_u$. During this traversal, we maintain a heap of the $g$-values of the colors encountered on the path from $u$ to the currently examined node. When we encounter the representative $v$ of a new $t_4$-block, $B_v$, we iterate through all the colors contained in $B_v$ or in $B_u$ and delete them from the heap. Then, $T[B_u][B_v]$ is assigned the color of maximal value of $g$, contained in the heap after the deletions. After retrieving the sought value, we insert the deleted elements back into the heap and continue the traversal. 

Each DFS traversal will visit each node exactly once, and will compute the values $T[B_u][B_v]$ for $n/t_4$ distinct blocks $B_v$, requiring $O(t_4\log n)$ time for each block. Thus, the time required for one DFS traversal is $O(n\log n)$. There will be $O(n/t_4)$ traversals in total, thus the runtime of the precomputation algorithm is $O(n\log n(n/t_4))=O(n(n/t_2))$. 

Finally, we note that the optimal color, from the set of colors $C_{path} \cap \overline{C(B_3(i))} \cap \overline{C(B_3(j))}$, is the color with the maximum value of the function $g$, among the colors in the $3$ examined sets.
\begin{align*}
    c_{opt}=\arg\max \{g(c_{1},u,v), g(c_{2},u,v), g(c_{3},u,v)\} \\
    c_1\in C_{path} \cap C(B_4(u))\cap C(B_4(v)) \cap \overline{C(B_3(u))} \cap \overline{C(B_3(v))}\\
    c_2\in C_{path} \cap C(B_3(v)) \cap \overline{C(B_4(u))} \cap \overline{C(B_4(v))}\\
    c_3 \in C_{path} \cap \overline{C(B_4(u))} \cap \overline{C(B_4(v))}
\end{align*}

The runtime of the procedure for computing the table $T$ required for subtask 1 is $O(n(n/t_2))$, which corresponds to the preprocessing time mentioned in the previous sections, in Lemma~\ref{lemma:path_g_max_value_parameter_t} and Theorem~\ref{theorem:path_g_max_value_linear_space}. 

\textbf{Solving subtasks 2 and 3.} To solve subtasks 2 and 3, we can directly apply the algorithm used in \textbf{Case 2} of the precomputation procedure for Subtask 1, to achieve an $O(n(n/t_3)\log\log n)\subset O(n(n/t_2))$ time algorithm.
 
\textbf{Solving subtask 5.} For subtask $5$, we will have to adapt the procedure associated with Lemma~\ref{lemma:universal_precomputation_traversal} to the multi-level partition hierarchy designed specifically for subtask $5$. For every blocking factor $s_k$, where $k>0$, from the list $s_0,s_1,\dots, s_{|s|}$, we will apply the procedure from Lemma~\ref{lemma:universal_precomputation_traversal} to the pair of blocking factors $(s_k, s_{k-1})$, and compute the data structure by iterating through every pair of blocks $(B_u, B_v)$, where $B_u$ is an $s_{k-1}$ block and $B_v$ is an $s_{k}$ block. This way, we obtain an algorithm with the following runtime:
\begin{align*}
    \sum_{k>0}n\frac{n}{s_k} = n^2/t_2 \sum_{k>0} \frac{1}{r_k} = O(n^2/t_2)
\end{align*}   
where $r_k=s_k/t_2$. 

\textbf{Solving subtasks $9$ and $6$.} The precomputation algorithms for subtasks $9$ and $6$ follows naturally from the previously described solutions, and can be designed by applying the algorithm from Lemma~\ref{lemma:universal_precomputation_traversal} adapted to pairs formed by a $t_k$ and a $t_{k'}$ block, respectively. These procedures take $O(n(n/t_2))$ time in total.

We remind that subtasks 4, 7, and 8 are symmetric to the other described subtasks, thus, they do not require a designated preprocessing algorithm.

The time requirement for any preprocessing algorithm assigned to the presented subtasks is $O(n(n/t_2))$. Thus, the overall preprocessing time for the whole data structure is $O(n(n/t_2))$. Also, note that the presented algorithms can be implemented such that they require $O(n)$ working space.

\end{document}